\documentclass[reprint, 
superscriptaddress,
amsmath,amssymb,
aps,
pra,
]{revtex4-2}

\usepackage{graphicx}
\usepackage{dcolumn}
\usepackage{bm}
\usepackage{array}
\usepackage{hyperref}

\usepackage{multirow}
\usepackage{amssymb}
\usepackage{amsmath}
\usepackage{amsthm}
\usepackage{mathrsfs}
\usepackage{amsfonts}
\usepackage{color}
\usepackage{xspace,soul}
\usepackage{tcolorbox}
\usepackage{physics}
\usepackage{siunitx}
\usepackage{float}
\usepackage[caption=false]{subfig}
\usepackage{ctable,textcomp,gensymb}

\fontfamily{mathptmx}
\newcolumntype{"}{@{\hskip\tabcolsep\vrule width 1pt\hskip\tabcolsep}}
\makeatletter
\newcommand{\thickhline}{%
    \noalign {\ifnum 0=`}\fi \hrule height 1pt
    \futurelet \reserved@a \@xhline
    }

\newtheorem{property}{Property}[section]

\newtheorem{theorem}{Theorem}[section]

\newtheorem{lemma}[theorem]{Lemma}
\theoremstyle{definition}

\usepackage[normalem]{ulem}

\begin{document}

\title{Quantum dynamics of atoms in number-theory-inspired potentials}

\author{D. Cassettari}
\affiliation{SUPA School of Physics $\&$ Astronomy, University of St Andrews, North Haugh, St Andrews KY16 9SS, UK}

\author{O. V. Marchukov}
\affiliation{Institute of Photonics, Leibniz University Hannover, Nienburger St. 17, 30167 Hannover, Germany}
\author{B. Carruthers}
\affiliation{SUPA School of Physics $\&$ Astronomy, University of St Andrews, North Haugh, St Andrews KY16 9SS, UK}

\author{H. Kendell}
\affiliation{Quantum Engineering Centre for Doctoral Training, University of Bristol, Bristol BS8 1FD, UK}
\affiliation{Quantum Engineering Technology Laboratories, H. H. Wills Physics Laboratory and Department of Electrical and Electronic Engineering, University of Bristol, Bristol BS8 1FD, UK}

\author{J. Ruhl}
\affiliation{Department of Physics, University of Massachusetts Boston, Boston Massachusetts 02125, USA}

\author{B.~De~Mitchell~Pierre}
\affiliation{Department of Mathematics, University of Massachusetts Boston, Boston Massachusetts 02125, USA}

\author{C.~Zara}
\affiliation{Department of Mathematics, University of Massachusetts Boston, Boston Massachusetts 02125, USA}

\author{C. A. Weidner}
\affiliation{Quantum Engineering Technology Laboratories, H. H. Wills Physics Laboratory and Department of Electrical and Electronic Engineering, University of Bristol, Bristol BS8 1FD, UK}

\author{A. Trombettoni}
\affiliation{Department of Physics, University of Trieste, Strada
  Costiera 11, I-34151 Trieste, Italy}
\affiliation{SISSA and INFN, Sezione di Trieste, Via Bonomea 265, I-34136 Trieste, Italy}

\author{M. Olshanii}
\affiliation{Department of Physics, University of Massachusetts Boston, Boston Massachusetts 02125, USA}

\author{G. Mussardo}
\affiliation{SISSA and INFN, Sezione di Trieste, Via Bonomea 265, I-34136 
Trieste, Italy}

\date{\today}

\begin{abstract}

In this paper we study transitions of atoms between energy levels of several number-theory-inspired trapping potentials under the effect of time-dependent perturbations. First, we simulate in detail the case of a trap whose single-particle spectrum is given by the prime numbers. We investigate one-body Rabi oscillations and the excitation lineshape for two resonantly coupled energy levels, and we show that quantum control is a faster method for state preparation than periodic perturbation. Next, we investigate cascades of such transitions, particularly whether one can construct a quantum system where the existence of a continuous resonant cascade from a given initial energy eigenstate is predicated by the validity of a given statement in number theory. We find that such resonance cascades, in a suitably-designed one-body system, can be used to verify that the sequence of natural numbers is closed under multiplication. We further present ideas for two more resonance cascade experiments designed to illustrate the validity of the Diophantus-Brahmagupta-Fibonacci identity and the validity of the Goldbach conjecture. 
\end{abstract}

\maketitle

\section{Introduction}

In recent years, optical trapping of ultracold atoms has advanced to the point that it is now possible to control, with precision, the profile of the trapping potential. This progress has been enabled by light sculpting techniques~\cite{yelin2021} relying on the use of spatial light modulators~\cite{cassettari2014}, digital micromirror devices~\cite{greiner2016, gross2016}, and fast-scanning acousto-optic deflectors~\cite{boshier2009}. 

These developments are particularly suited to address, in a realistic and tunable setup, suggestions and ideas at the 
interface between quantum physics and number theory~\cite{GMscattering,Hutchinson2011,Wolf2020}. To this end there has recently been experimental progress in creating \emph{prime number potentials} for cold atoms~\cite{cassettari2022_220203446}. These are trapping potentials whose first $N_{\text{b}}$ eigenenergies are proportional to the first $N_{\text{b}}$ prime numbers. The rationale is that such highly controllable potentials, which can in general be engineered to have an assigned sequence of numbers as their energy levels, are relevant for a host of theoretical ideas. For instance, they can be used for the implementation of primality tests~\cite{GMscattering}, prime factorization of integers (see Refs. \cite{gleisberg2018_035009,mussardo2023} and references therein), and for studying statements from number theory such as the Goldbach conjecture, which states that every even natural number greater than two can be written as sum of two prime numbers. Of course, one cannot prove mathematical theorems by using a physical device in which a finite (even very large) number of prime numbers is encoded -- but, on the other hand, one can systematically extract predictions from number theory theorems and conjectures and design experiments based on them (possibly looking for counterexamples for conjectures). At the same time, by cross-fertilization, the physical realization of number-theory-inspired setups can generate new ideas to tackle number theory problems.

In this paper, we consider one-dimensional trapping potentials where the energy levels of the system have been engineered to produce a desired spectrum, and we study transitions between energy levels under the effect of time-dependent perturbations. We consider two complementary scenarios. First, in Sec. \ref{parametricresonance} we study transitions in a potential with energy levels proportional to the first $N_{\text{b}}$ prime numbers. As the prime numbers are not regularly spaced, a time-dependent perturbation on resonance between the ground state and a target excited state transfers population between these states, while the population of other states remains negligible. We also explore how quantum control protocols can speed up this state preparation. We remind that the selective excitation of just one state is a critical state preparation step prior to carrying out factorization algorithms described in Refs.~\cite{mussardo2023,gleisberg2015}. 

The next sections of the paper are dedicated to further exploration of the use of spectrally-engineered systems. In Secs. \ref{cascades} and \ref{s:natural_set} we introduce \emph{resonance cascades}, where a particle prepared in a given highly excited state is induced to cascade down in energy due to an external perturbation. We use these resonance cascades as a paradigm to study a set of statements from number theory. As a test bench, we consider a potential with an energy spectrum given by the logarithms of the natural numbers. Here we find that, because the sequence of natural numbers is closed under multiplication, a perturbation at an appropriate frequency can generate a resonance cascade (or ladder), in which many levels are populated. Conversely, if we remove from the spectrum the logarithm of one of the natural numbers, this resonance cascade is inhibited.

Finally, in Sec. \ref{futureideas} we outline future ideas on how resonance cascades can be used to study the Diophantus-Brahmagupta-Fibonacci identity~\cite{dudley_Number_Theory1970} (Lemma 1, p.\ 142) and the Goldbach conjecture~\cite{dudley_Number_Theory1970} (p.\ 147). In the first case, the potential of interest has energy levels given by the logarithm of the sum of two integers squared, whereas the second case is based on the prime number potential investigated in Section \ref{parametricresonance}, but where a cascade is enabled if the potential contains two weakly-interacting particles. Thus, this work outlines and studies in detail the necessary ingredients for multiple experimental verifications of number theoretic statements, from the spectral engineering to the initial state preparation and driving required to demonstrate, via resonance cascades, the identities and conjectures under study.

\section{A resonance in the prime number potential \label{parametricresonance}}

In this section, we discuss the relevant state preparation required for the subsequent number theoretic experiments. We explore this in two ways: first through perturbative driving on resonance, where we recover behavior similar to Rabi oscillations between the interrogated states. Secondly, we explore the potential of quantum control to speed up state preparation, showing that the state transfer can be sped up by about a factor of five.

\subsection{Periodic perturbations}

Let us consider the one-dimensional Hamiltonian:
\begin{equation}
\hat{H}^{\text{P}}= \frac{\hat{p}^2}{2m} + U^{\text{P}}(x)
\,\,,
\label{1d_Sch}
\end{equation}
where $U^{\text{P}}(x)$ is a potential having as assigned spectrum the first $N_{\text{b}}$ prime numbers times an energy scale, denoted by $U_0$, depending on the physical realization of the system. To be more precise, we could use the notation $U_{N_\text{b}}^{\text{P}}(x)$, but for the sake of simplicity we will omit the index $N_\text{b}$. The potential $U^{\text{P}}(x)$ is referred to as the {\it prime number potential}.

Fig. \ref{primePot} shows $U^{\text{P}}(x)$ with eigenvalues given by the first $N_{\text{b}}=20$ prime numbers:
\begin{align}
\begin{split}
&
E^{\text{P}}_{n}  = U_{0} p_{n} 
\\
&
p_{1},\,  p_{2},\, p_{3},\, \ldots p_{20} = 2,\, 3,\, 5,\, 7,\,\ldots 71.
\end{split}
\end{align}
This potential has been calculated using the methods of supersymmetric quantum mechanics and has been optically implemented with a spatial light modulator \cite{cassettari2022_220203446}. In general, the methods used to calculate this potential can be used to calculate potentials with arbitrary (and finite) spectra. Notice that $U^{\text{P}}(x)$ is going to a constant as $x \to \infty$, so that above the eigenvalue $p_{20}$ the continuous part of the spectrum starts. From now on, we measure energies in units of the energy scale $U_0$ and lengths in units of the quantity $a=\hbar/\sqrt{m U_0}$.

%

\begin{figure}[t]
\begin{center}
\includegraphics[width=0.9\columnwidth]{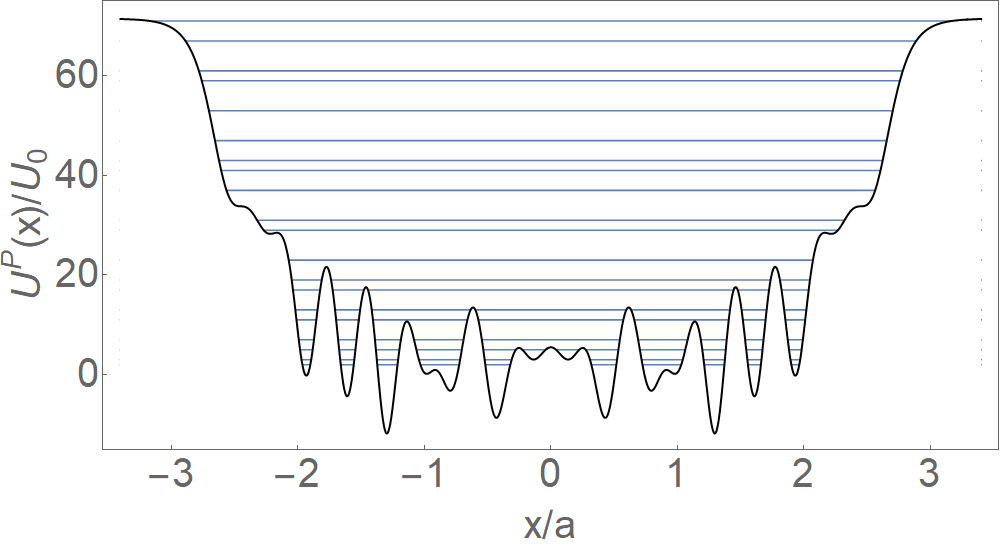}
\caption{Prime number potential with 20 energy levels corresponding to the first 20 prime numbers. Here, $a=\hbar/\sqrt{mU_{0}}$ is the characteristic length scale of the potential (with $m$ the particle mass), and $U_0$ is the corresponding energy scale.}
\label{primePot}
\end{center}
\end{figure}

\begin{figure}[b]
\begin{center}
\includegraphics[width=1\columnwidth]{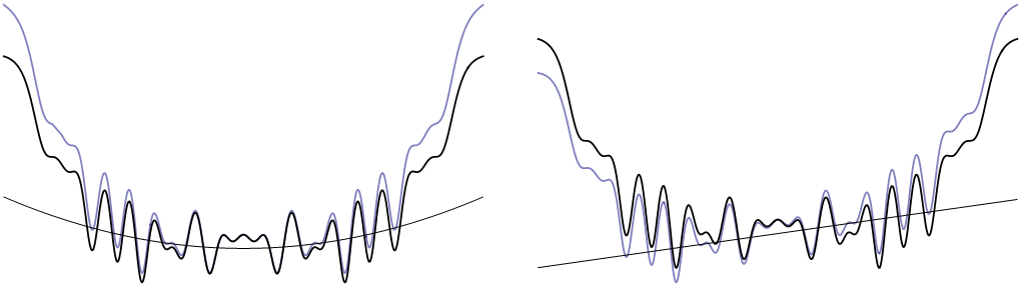}
\caption{Even (left) and odd (right) perturbations, where $\beta=2$ for the even case and $\beta=4$ for the odd case. These values of $\beta$ are significantly larger than those used in the simulations to better visualize the effect of the perturbations on the potential. The axes on these plots are the same as in Fig.~\ref{primePot}.}
\label{even_odd_pert}
\end{center}
\end{figure}

\begin{figure}[t]
    \subfloat[]{\includegraphics[width=0.9\columnwidth]{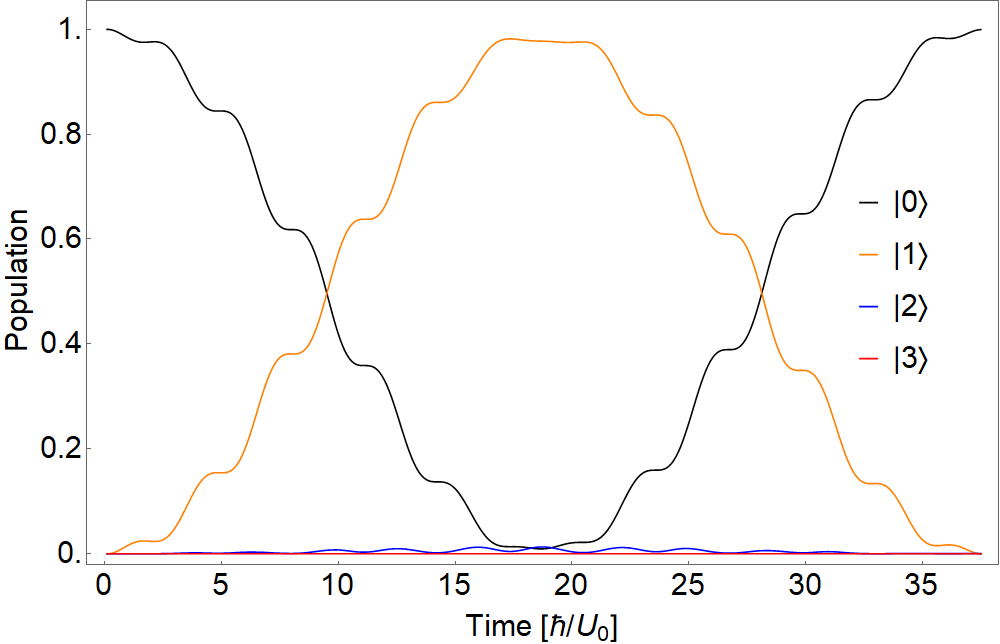}}
    \newline
    \subfloat[]{\includegraphics[width=0.9\columnwidth]{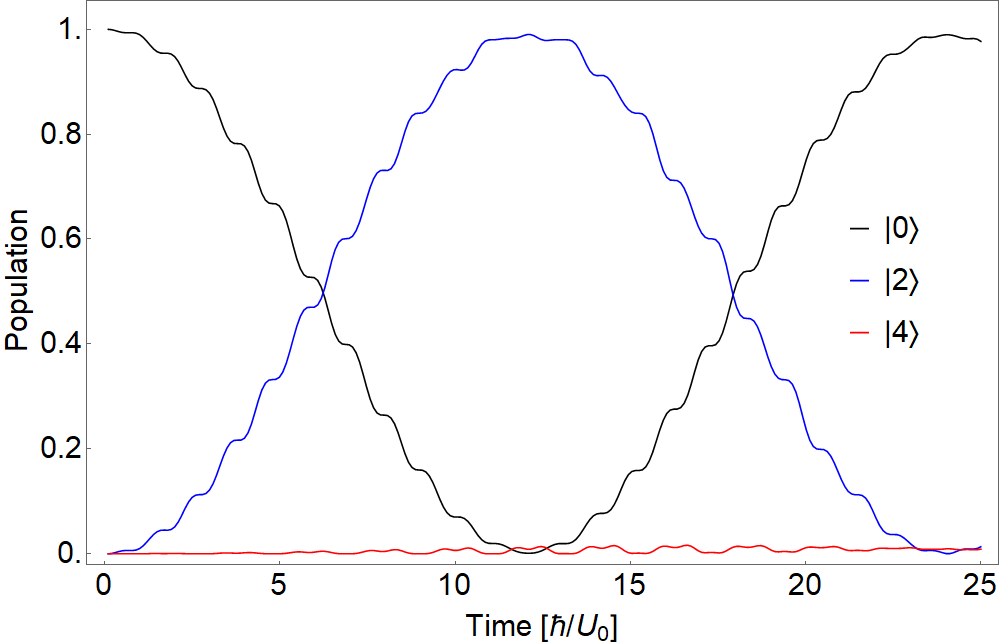}}
    \newline
    \caption{(a) Odd transition between the ground state and first excited state with $\beta=0.25$, where a small population of the second excited state can be seen. (b) Even transition between the ground state and second excited state with $\beta=0.5$. This driving strength also leads to a small but visible population of the fourth excited state (the next transition of even parity).}
\    \label{Rabi_osc}
\end{figure}

In this section we simulate the transfer between bound states of this prime number potential by applying time-dependent perturbations of the form
\begin{equation}
V(x,t)=f(x) \cos(\omega t)
\end{equation}
where we examine two cases for $f(x)$ (see Fig. \ref{even_odd_pert}): $f_\mathrm{odd}(x)=\beta (x/a)$ and $f_\mathrm{even}(x)=\beta (x/a)^2$, which induce transitions between states of opposite parity and between states of the same parity, respectively. Specifically, $f_\mathrm{odd}(x)$ enables the transfer between the ground state $|0\rangle$ and the first-excited state $|1\rangle$, while $f_\mathrm{even}(x)$ enables the transfer between $|0\rangle$ and the second-excited state $|2\rangle$. In both cases, the perturbation is close to resonance with the transition ($\omega=U_0/\hbar$ and $\omega=3U_0/\hbar$ for the first and second excited states, respectively), and the anharmonicity of the potential ensures that the population of higher states remains small, so that we have essentially isolated a two-level system. Experimentally, these perturbations can be realized with additional optical or magnetic fields. Fig. \ref{Rabi_osc} shows Rabi oscillations between $|0\rangle$ and $|1\rangle$ (a) and between $|0\rangle$ and the $|2\rangle$ (b), both obtained by numerical integration of the time-dependent Schr\"odinger equation. The Rabi frequency $\Omega_\mathrm{R}$ is given by the matrix element of the perturbation between the ground state and the desired excited state \cite{cjfoot}. The $\beta$ values are chosen according to the following criterion: a larger $\beta$ would lead to a higher Rabi frequency and a faster transfer to the desired excited state, which is experimentally advantageous, however this would come at the price of reduced efficiency of the transfer to the desired excited state, because a larger $\beta$ would increase the population of higher states. In this respect, Fig. \ref{Rabi_osc} shows the fastest possible transfer while achieving a good fidelity corresponding to $>98\%$ of the population in the desired excited state. Note that in addition to the Rabi oscillations, the time evolution also displays fast ripples at frequency $2\omega$. These are due to the fact that the rotating wave approximation is not perfectly satisfied for the perturbation strengths we are using.

\begin{figure}[t]
    \subfloat[]{\includegraphics[width=0.9\columnwidth]{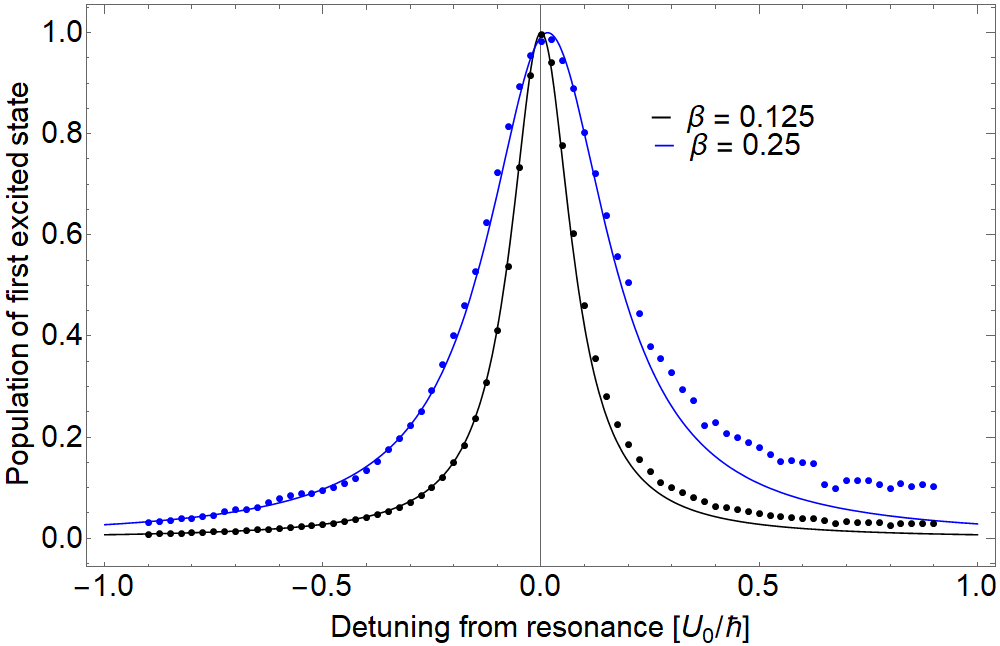}}
    \newline
    \subfloat[]{\includegraphics[width=0.9\columnwidth]{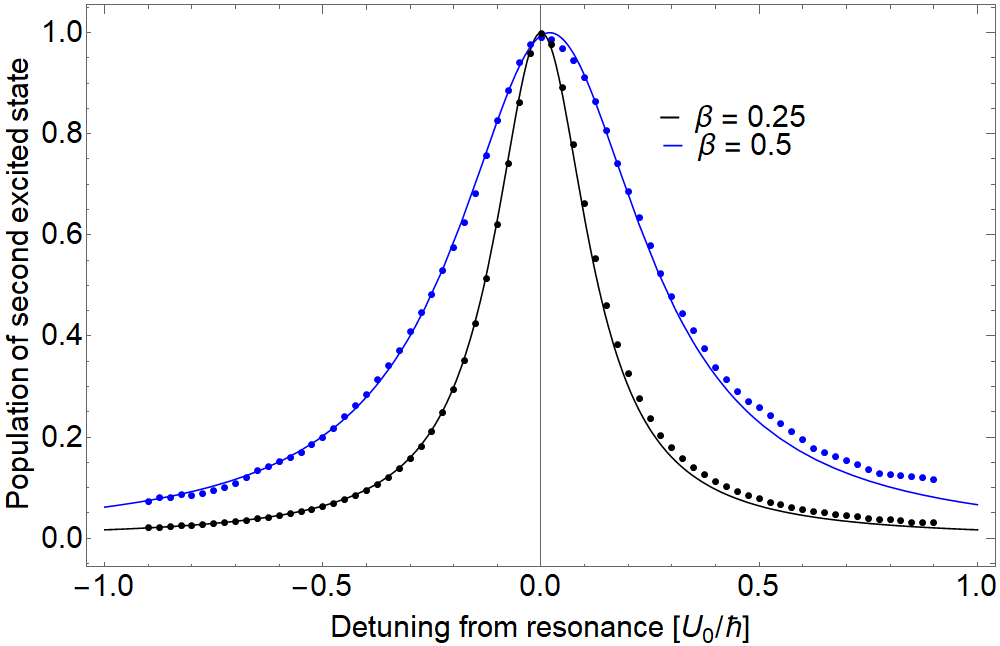}}
    \newline
    \caption{(a) Lineshape of the odd transition for $\beta$=0.125 and $\beta$=0.25. (b) Lineshape of the even transition for $\beta$=0.25 and $\beta$=0.5. Numerical data are displayed alongside Eq.~\eqref{Lorentzian} with their respective value of $\Omega_R$. We note that for the same value $\beta$=0.25, the width of the odd transition is larger than the width of the even transition. This is due to the different couplings, as $\langle 0|x|1\rangle=0.68$ and $\langle 0|x^2|2\rangle=0.52$.}
\    \label{detunings}
\end{figure}

To further characterize the system, we study the effect of off-resonant perturbations. For a two-level system in the rotating wave approximation, the maximum population of the excited state is
\begin{equation}
    P_{ex, max}=\frac{\Omega_R^2}{\Omega_R^2+\delta^2}\text{ .}
\label{Lorentzian}
\end{equation}
where $\delta$ is the detuning from resonance~\cite{cjfoot}. The maximum population is 1 for resonant perturbation ($\delta=0$), and displays a Lorentzian lineshape for finite detunings with FWHM=$2\Omega_R$. These Lorentzian lineshapes are found numerically in our system for both the odd and the even transition, as shown in Fig. \ref{detunings}. The good agreement between the numerical data and Eq. \eqref{Lorentzian} (with the value of $\Omega_R$ calculated from the relevant matrix element of the perturbation) shows that our system behaves as a two-level system to a good approximation. We note that for larger values of $\beta$, the center of the transition slightly shifts to the right. We attribute this to the Bloch-Siegert effect~\cite{BlochSiegert}, which occurs beyond the rotating wave approximation. Similarly, the deviation of the numerical data from the Lorentzian, visible at positive detunings, is also likely due to the rotating wave approximation not being well satisfied. 

\begin{figure}[t]
    \subfloat[$\beta (x/a)^2$] {\includegraphics[width=\columnwidth]{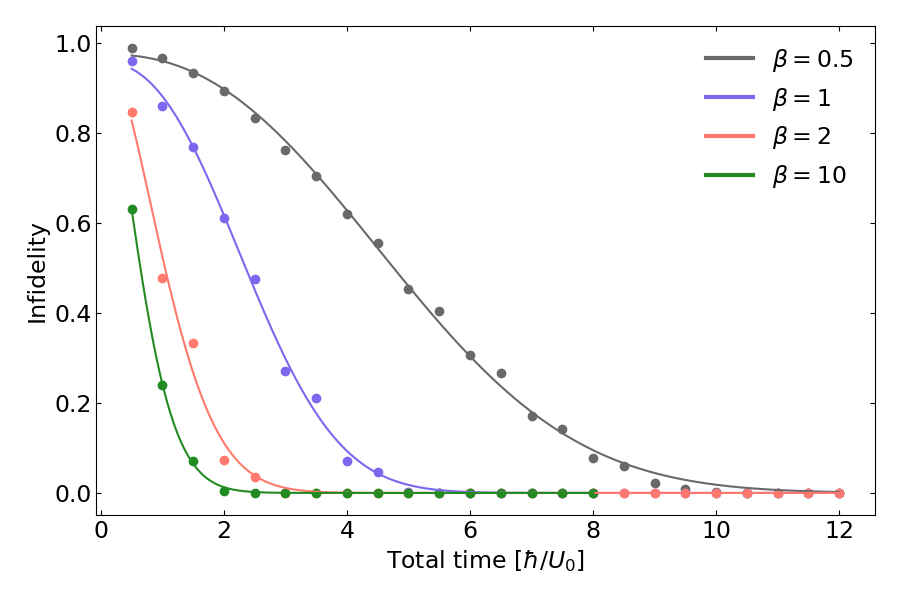}}
    \newline
    \subfloat[$\beta (x/a)$]
    {\includegraphics[width=\columnwidth]{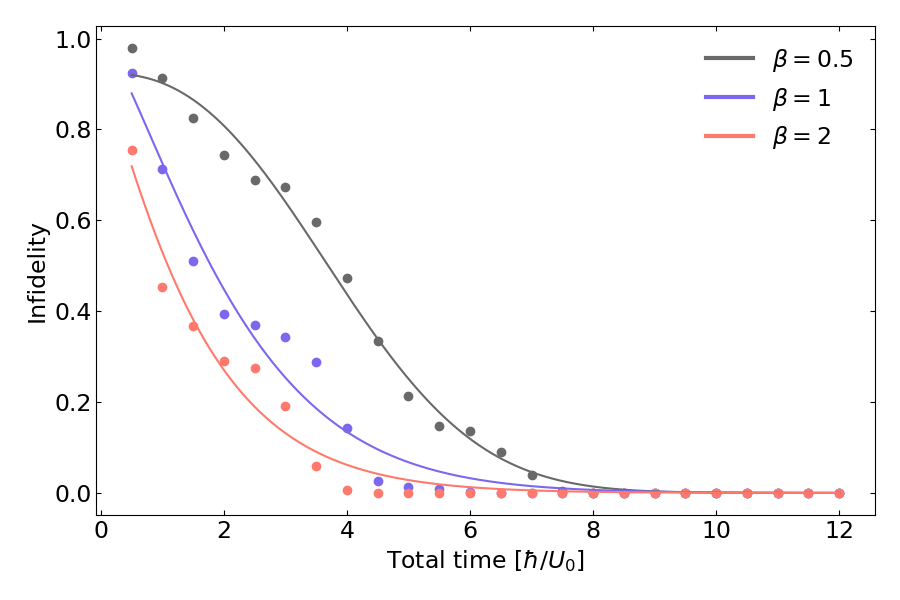}}
    \newline
    \caption{Infidelity vs. time curves for (a) $\ket{0}\rightarrow \ket{2}$ transfer with $\beta (x/a)^2$ modulation and (b) $\ket{0}\rightarrow \ket{1}$ transfer with $\beta (x/a)$ modulation. The strength of the $\beta$ value are labelled in the plot legends. The individual data points are the best results of $2000$ optimizations, and the line is a guide to the eye.}
    \label{fig:FT_curve}
\end{figure}

In an experimental implementation with ultracold atoms (for instance, in rubidium), the typical resonance frequencies will be $\sim 1$~Hz and typical transfer times to the excited state will be on the order of $10$~s. These slow timescales are because the characteristic length scale of the potential $a$ needs to be well above the optical resolution of the optical system that produces the potential. This constraint on $a$ puts an upper limit on $U_0$ via $a=\hbar/\sqrt{mU_{0}}$, and consequently a lower limit on the time unit $\hbar/U_{0}$ \cite{cassettari2022_220203446}. 
Higher-resolution projection systems providing smaller $a$ values (as in quantum gas microscopes~\cite{greiner2009,Bloch2011}) allows for faster state transfers. However, even with modest values of $a$, we can speed up state transfer via optimal control methods, as we show next.

\subsection{\label{ss:QC}Quantum control}

\begin{figure}[t]
    \subfloat[$\beta (x/a)^2$] {\includegraphics[width=0.9\columnwidth]{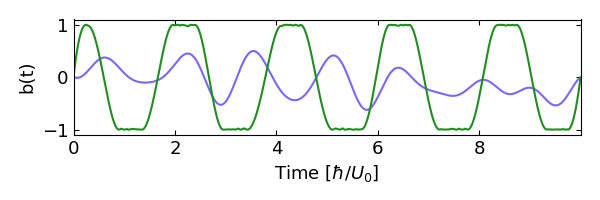}}
    \newline
    \subfloat[$\beta (x/a)$]
    {\includegraphics[width=0.9\columnwidth]{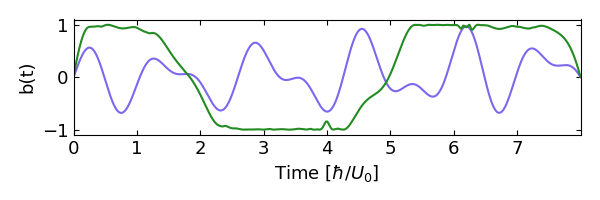}}
    \newline
    \caption{Optimal controls $b(t)$ for (a) $\beta(x/a)^2, \ket{0}\rightarrow \ket{2}$ transfer and (b) $\beta(x/a), \ket{0}\rightarrow \ket{1}$ transfer. Here $\beta=0.5$. The blue (green) line denotes the initial (optimized) control.}
    \label{fig:controls}
\end{figure}

In this section, we explore two state transfer protocols simulated using the QEngine C++ library~\cite{Sherson_2019} and optimized using the GRAPE algorithm~\cite{Khaneja2005} implemented within QEngine:
\begin{itemize}
    \item $\beta (x/a)^2$ modulation for $\ket{0}\rightarrow \ket{2}$ transfer,
    \item $\beta (x/a)$ modulation for $\ket{0}\rightarrow \ket{1}$ and $\ket{0}\rightarrow \ket{3}$ transfer.
\end{itemize}
For each case, we implement a control of the form $\beta\, b(t) (x/a)^n$ with $n = 1, 2$ for the linear and quadratic cases, respectively. The control $b(t)$ is optimized for a given static value of $\beta$, and $-1\leq b(t)\leq 1\, \forall \, t$.

These protocols provide a benchmark that serves to compare optimal control results with resonant driving results, as well as demonstrating an interesting state preparation proof-of-principle that can be useful for, e.g., the resonant cascade protocols described in later sections of the paper. For the optimization, the initial state is the ground state of the prime number potential, and we use regularization~\cite{Winckel_2008, Sherson_2018} to limit the bandwidth of the control in order to preserve experimental viability.

Using the unit scalings provided, the best results are tabulated in Table~\ref{tab:results}, and plots of the fidelity vs. control time for $\beta (x/a)^2$ and $\beta (x/a)$ are shown in Fig.~\ref{fig:FT_curve}. These show the minimum time required for a control to reach $0.99$ fidelity, where fidelity is defined as $\mathcal{F} = |\bra{\psi_\mathrm{des}}\ket{\psi_\mathrm{opt}}|^2$, with $\ket{\psi_\mathrm{des}}$ corresponding to the either $\ket{1}, \ket{2}$, or $\ket{3}$ and $\ket{\psi_\mathrm{opt}}$ the final state after optimization. The infidelity is $\mathcal{I} = 1-\mathcal{F}$. The optimal controls (i.e., those found at the minimum time where $\mathcal{F} > 0.99$) are plotted in Fig.~\ref{fig:controls} for $\beta = 0.5$.
\\
\begin{table}[]
\begin{tabular}{|l"c|c|c|}

\hline
Transfer Time $(\hbar/U_0)$           & $\ket{1}$     & $\ket{2}$     & $\ket{3}$     \\ \thickhline
$\beta=0.5, n=2$             & -             & 9.5           & -             \\ \hline
$\beta=1, n=2$                 & -             & 5.0           & -             \\ \hline
$\beta=2, n=2$               & -             & 3.0           & -             \\ \hline
$\beta=10, n=2$              & -             & 2.0           & -             \\ \thickhline
$\beta=0.5, n=1$               & 7.5           & -             & $\gg12.0$          \\ \hline
$\beta=1, n=1$                   & 5.5           & -             & 9.0           \\ \hline
$\beta=2, n=1$                 & 4.0           & -             & 7.0           \\ \hline
\end{tabular}

\caption{\label{tab:results} Optimized transfer times ($\mathcal{F}>0.99$) from the ground state for a variety of protocols. Note that the case shown as $\gg\SI{12}{\hbar/U_0}$ has not been optimized for longer times due to computer cluster limitations.}
\end{table}

The results in Table~\ref{tab:results} show that quantum control enables the transfer to the desired excited state with fidelities comparable to the case of the periodic perturbations, but with transfer times reduced by up to a factor $\sim5$. Differently from the case of periodic perturbation, with quantum control it is possible to increase $\beta$, hence reducing significantly the transfer time, while keeping the population of higher levels negligible. All the transfer times reported on Table~\ref{tab:results} correspond to experimental times of 5~s or less, which is feasible for state preparation prior to the experiments discussed in the following sections.

\section{Resonance cascades}\label{cascades}
In this section we introduce the idea of using {\it resonance cascades} to study statements from number theory. That is, given a number-theoretical statement, we will first recast the statement or some of its corollaries as an assertion of an existence of an uninterrupted arithmetic progression. This arithmetic progression is a subsequence of a larger sequence of numbers (the principal sequence). The latter sequence will constitute the energy spectrum of a quantum system. 
The step between elements of the arithmetic subsequence in question will correspond to the frequency of a time-dependent perturbation. If the arithmetic progression is fully contained in the principal sequence, the corresponding quantum system will be able to travel along the energy axis, exclusively via resonant transitions.

\begin{table*}
\centering
\begin{tabular}{|>{\centering}m{4.5cm}|>{\centering}m{2.1cm}|>{\centering}m{3cm}|>{\centering}m{2.5cm}|>{\centering}m{2.5cm}|>{\centering}m{2cm}|p{0cm}}
\hline
Number theoretic statement           & Corollary to be studied     & Principal sequence of numbers     & Arithmetic progression & Corresponding quantum spectrum & Excitation frequency ($\hbar\Omega$) & \\ \thickhline
(I) The sequence of natural numbers is closed under multiplication.
& 
Given $\tilde{n}\in\mathbb N$,\\$\forall m \in \mathbb{N}_0$,\\$\tilde{n}^m\in\mathbb{N}$ &  
$\ln{(n)}$,\\$n\in\mathbb{N}$. & 
$m\ln{(\tilde{n})}$,\\$\tilde{n}\in\mathbb{N}$,\\$m \in \mathbb{N}_0$ &
single-particle,\\$E_n^{\text{L}} = U_0 \ln{(n)}$,\\$n\in\mathbb{N}$ &
$U_0 \ln{(\tilde{n})}$ &
\\ \hline
(II) The sequence $S$ of sums of two squares of integers is closed under multiplication. &
Given $\tilde{s}\in S$,\\$ \forall m \in \mathbb{N}_0$,\\$\tilde{s}^m \in S$. &
$\ln{(s_n)}$ with \\ $s_n \in S$, \\ $s_n \neq 0$ &
$m \ln(\tilde{s})$,\\$m\in\mathbb{N}_0$, and $\tilde{s}\in S$,\\$\tilde{s}\neq 0$ & 
single particle, \\ $E_n^{\text{L2}} = U_0\ln{(s_n)}$,\\$s_n\in S$,\\$s_n\neq 0$ &
$U_0 \ln{(\tilde{s})}$ &
\\ \hline
(III) Any even number greater than 2 is a sum of two primes.&
Any even number greater than 4 is the sum of two primes greater than 2. &
$w_{n,\,l}$ with 
$w_{1}=3+3=6$,\\ $w_{2}=3+5=8$,\\ $w_{3,\,1}=5+5=10$,\\ $w_{3,\,2}=3+7=10$,\\ $w_{4}=5+7=12,\,\,\ldots$ &
 $2m$ with $m=3,\,4,\,5,\,\ldots$ & 
two interacting particles, each with $E_{n} = U_{0} p_{n}$ where $p_n$ are primes $>2$, i.e.\ $3,\,5,\,7,\,11\,\ldots$ &
$2U_{0}$ & \\
\hline
\end{tabular}
\caption{\label{tab:number-theory} Number theoretic problems to be studied in this work, where $\mathbb{N} = 1, 2, 3, \ldots$ is the sequence of natural numbers, $\mathbb{N}_0 = 0, 1, 2, \ldots$ is the sequence of whole numbers (i.e., natural numbers including zero). In (II) we do not consider $0^2 + 0^2 = 0$ for technical reasons. Note the degeneracies in the principal sequence of numbers $w_{n,l}$ considered in (III). Preliminary numerical results for case (I) are published in Ref.~\cite{marchukov2024}.}
\end{table*}

The following three examples are considered in this paper (also summarized in Table~\ref{tab:number-theory}):

\vspace{0.1cm}
\noindent \textbf{(I) The number-theoretic statement:} {\it The sequence of natural numbers is closed under multiplication}.

\noindent \textbf{Corollary to be studied:} Given a natural number $\tilde{n}$, any power of it, $(\tilde{n})^m$, will be present in the sequence of natural numbers.

\noindent \textbf{The principal sequence of numbers:} $\ln(n)$ with $n=1,\,2,\,\,3\,\ldots$. 

\noindent \textbf{The arithmetic progression:} $m \ln(\tilde{n})$ with $m=0,\,1,\,2,\,3,\,\ldots$, with
$\tilde{n}$ being a natural number.       

\noindent \textbf{The corresponding quantum system:} One particle in a trap whose energy spectrum is proportional to the logarithms of natural numbers, $E_{n}^{\text{L}} = U_{0} \ln(n),\,\, n=1,\,2\,,3\,,\ldots$. The excitation frequency will be given by $\hbar \Omega = U_{0} \ln{\tilde{n}}$. 

\vspace{0.1cm}
\noindent \textbf{(II) The number-theoretic statement:} {\it The sequence of sums of two squares of integers is closed under multiplication} (Diophantus-Brahmagupta-Fibonacci identity \cite{dudley_Number_Theory1970} (Lemma 1, p.\ 142)).

\noindent \textbf{Corollary to be studied:} Given a sum of two squares $\tilde{s}$, any power of it, $(\tilde{s})^m$, will be present in the sequence of sums of two squares.

\noindent \textbf{The principal sequence of numbers:} $\ln(s_{n})$ with $s_{1}=0^2+1^2=1,\,\,s_{2}=1^2+1^2=2,\,\,,s_{3}=0^2+2^2=4,\,\,s_{4}=1^2+2^2=5,\,\,s_{5}=2^2+2^2=9,\,\,\ldots$.\footnote{We exclude $0^2+0^2=0$ for technical reasons.} 

\noindent \textbf{The arithmetic progression:} $m \ln(\tilde{s})$ with $m=0,\,1,\,2,\,3,\,\ldots$, and
$\tilde{s}$ being a nonzero sum of two squares.  

\noindent \textbf{The corresponding quantum system:} One particle in a trap whose energy spectrum is proportional to the logarithms of sums of two  squares numbers, $E_{n}^{\text{L2}} = U_{0} \ln(s_n),\,\, n=1,\,2\,,3\,,\ldots$. The excitation frequency will be given by $\hbar \Omega = U_{0} \ln{\tilde{s}}$. 

\vspace{0.1cm}
\noindent \textbf{(III) The number-theoretic statement:} {\it Any even number greater than 2 is a sum of two primes} (Goldbach conjecture \cite{dudley_Number_Theory1970} (p.\ 147)).

\noindent \textbf{Corollary to be studied:} Any even number greater than 4 is a sum of two primes greater than 2. 

\noindent \textbf{The principal sequence of numbers, with degeneracies:} $w_{n,\,l}$ with 
$w_{1}=3+3=6,\,\,w_{2}=3+5=8\,\,,w_{3,\,1}=5+5=10,\,\,w_{3,\,2}=3+7=10,\,\,w_{4}=5+7=12,\,\,\ldots$. 

\noindent \textbf{The arithmetic progression:} $m\times 2$ with $m=3,\,4,\,5,\,\ldots$.  

\noindent \textbf{The corresponding quantum system:} Two interacting  particles in a trap whose one-body energy spectrum, 
$E^{\text{P}}_{n} = U_{0} p_{n},\,\,n=2,\,3,\,4,\,5,\,\,\ldots$ is proportional to the prime numbers greater than 2, i.e.\ $3,\,5,\,7,\,11\,\ldots$. The excitation frequency will be given by $\hbar \Omega = 2U_{0}$. 
%

\vspace{0.1cm}


Examples (I) and (II) require a trapping potential where the energy levels are logarithms of natural numbers and logarithms of the sequence of sums of two squares, respectively, as shown in Fig.~\ref{f:potentials_CUMULATIVE}. 
Example (III) requires two weakly interacting particles trapped in the prime number potential in Fig.~\ref{primePot}. Example (I), which is to be regarded as a proof-of-principle, is studied in depth in Section \ref{s:natural_set}. For examples (II) and (III), we outline ideas for future experiments in Section \ref{futureideas}.


\section{Resonance cascade predicated on the closeness of the natural set under multiplication \label{s:natural_set}}
%
\begin{figure*}[t]
\begin{center}
\includegraphics[width=\textwidth]{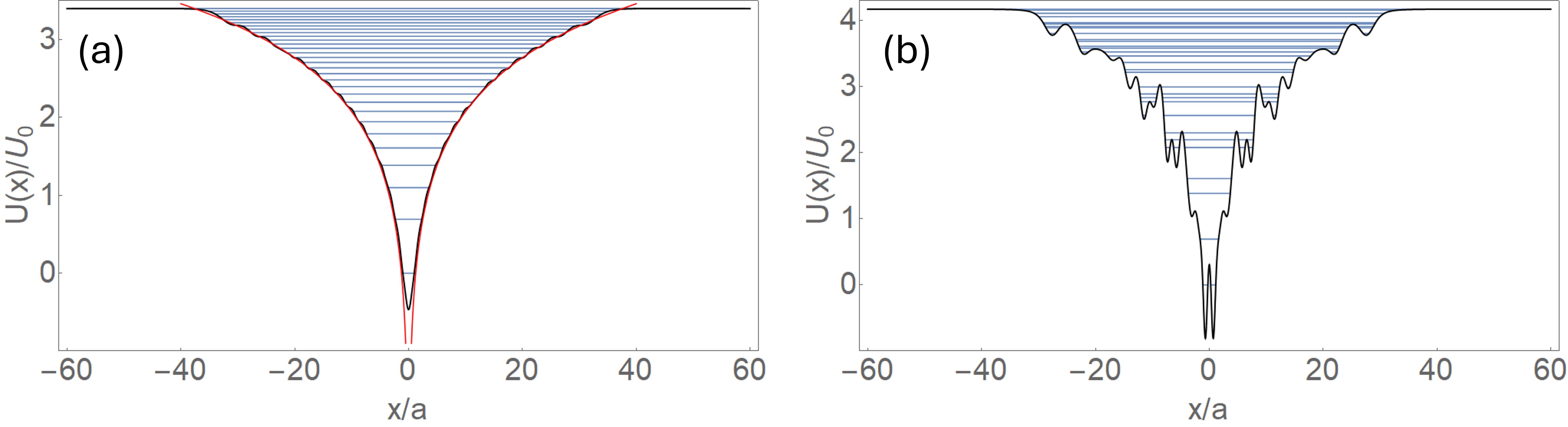}
\end{center}
\caption{Two potentials with a prescribed spectrum. (a) A potential whose lowest 30 energy levels are proportional to the logarithms of the first 30 natural numbers, as considered in Example (I). The red curve is its classical counterpart, Eq.~\eqref{log_n_potential_classical}.
(b) A potential where the lowest 30 energy levels are proportional to the logarithms of the sums of two squares, considered in Example (II).}
\label{f:potentials_CUMULATIVE}
\end{figure*}

\subsection{The set-up}
Consider a one-body potential $U^{\text{L}}(x)$ whose lowest $N_{\text{b}}$ bound state energies are given by 
\begin{align}
\begin{split}
&
E^{\text{L}}_{n} = U_{0} \ln n
\\
&
n=1,\, 2,\, 3,\,\ldots,\, N_{\text{b}}
\end{split}
\label{log_n}
\,\,.
\end{align}
Figure \ref{f:potentials_CUMULATIVE}(a) shows this potential for $N_{\text{b}}=30$. Similarly to the prime number potential of Section \ref{parametricresonance}, this potential is designed using supersymmetric methods \cite{cassettari2022_220203446}. As $N_{\text{b}}$ increases, the potential converges 
to a smooth classical limit. This classical counterpart can be extracted using procedures developed in classical mechanics (see \S 12 in \cite{book_landau_mechanics}). There exists a variety of classical potentials that can approximate our quantum potential; one of these is a logarithmic potential,
\begin{align}
U^{\text{L}}_{\text{cl.}}(x) = U_{0} \ln(\sqrt{\frac{2}{\pi}} \frac{x}{a})
\,\,,
\label{log_n_potential_classical}
\end{align}
where $a\, \equiv \hbar/\sqrt{mU_{0}}$ (see Appendix \ref{ss:logarithmic_potential}).
 
Note that the existence of the classical limit of an artificially constructed potential with an \emph{a priori} prescribed spectrum is not guaranteed. Figure~\ref{f:potentials_CUMULATIVE}(b) shows 
a potential whose first $N_{\text{b}}=30$ energy levels are proportional to the first 30 integers that are equal to the sum of two squares. We see that as $N_{\text{b}}$ increases, this potential retains its substantially quantum nature in that the small potential wells in Fig.~\ref{f:potentials_CUMULATIVE}(b) can be shown to contain only one or two eigenstates. 
\begin{figure*}[t]
\begin{center}
\includegraphics[width=1.\textwidth]{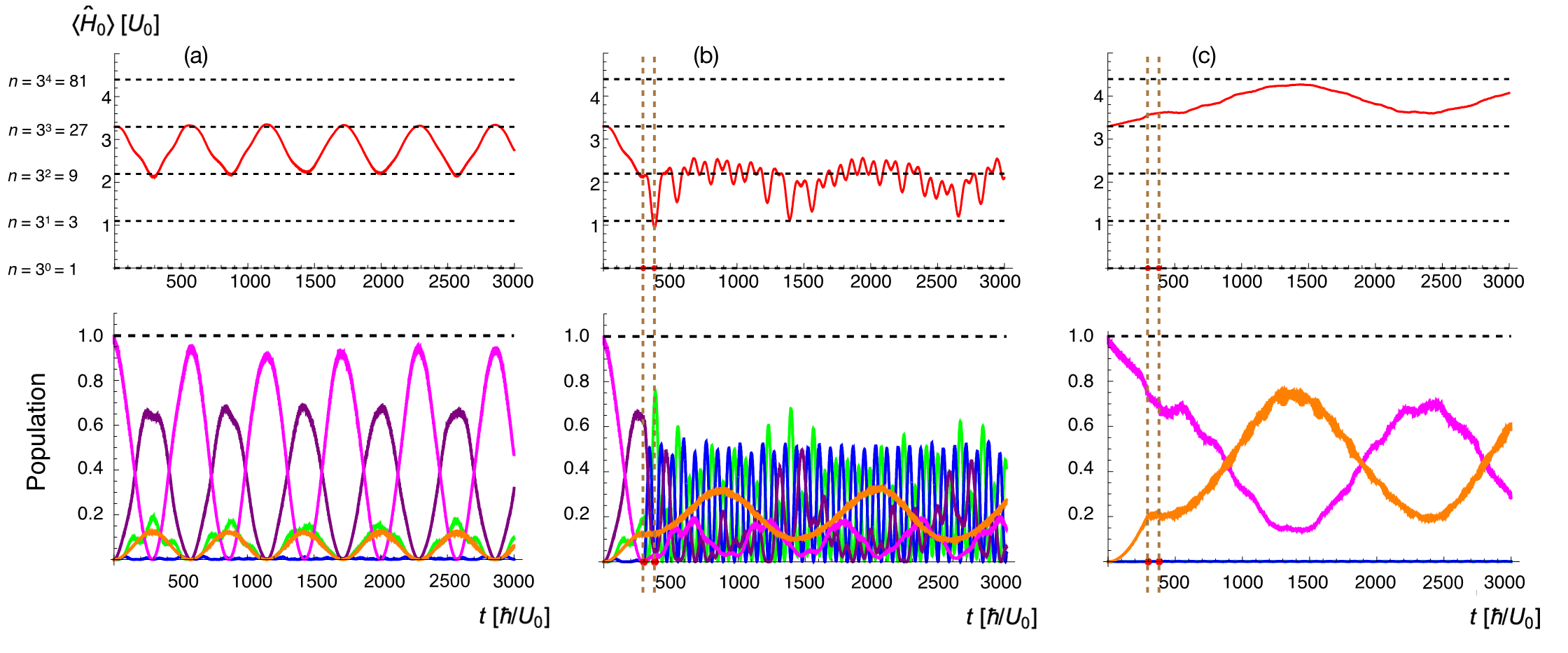}
\end{center} 
\vspace{-0.7cm}
\caption{A numerical experiment designed to verify if all powers of $3$ are present in the set of natural numbers. A single atom is initially placed in the $|n=27\rangle$ state of the potential $U^{\text{L}}(x)$ whose spectrum is given by 
$E_{n}^{\text{L}} = U_{0} \ln(n)$ ($n=1,\,2,\,3,\,\ldots,\,120$). The atom 
is subjected to a periodic perturbation of the form $V(x,\,t) = \beta U^\text{L}(x) \cos(\Omega t)$, with
$\beta =0.3$, and $\Omega = \frac{1}{\hbar} U_{0} \times \ln(\tilde{n})$, where $\tilde{n} = 3$. The upper panel shows the unperturbed energy of the system. The lower panel shows the population of the unperturbed eigenstates constituting the resonant cascade: $|n=1\rangle$ (green), $|n=3\rangle$ (blue), $|n=9\rangle$ (purple), $|n=27\rangle$ (magenta), and $|n=81\rangle$ (orange); in all three cases and for all reported instances of time, these populations comprise no less than 80\% of the total atomic population.  
(a) Mobility is impeded due to ``dark state'' localization (App.\ \ref{ss:dark_states}). Indeed, the lower panel shows that the time evolution is close to Rabi oscillations between $|n=27\rangle$ and $|n=9\rangle$. 
(b) In this numerical experiment, we introduce two ``windows of silence'' of a duration $\Delta t = \frac{\pi}{\Omega}$, 
at $t = \{298.83,\, 380.31\} \times \hbar/U_{0}$ (red dots on the abscissa and dashed vertical lines), during which the perturbation is stopped and then restarted after a quarter of the perturbation period, at the phase it had just before the window. The windows are designed to break the ``dark states'' and in doing so, promote mobility along the energy axis.  The level population indeed  starts moving down in energy while remaining within the set of resonantly coupled energy levels $U_{0} \times \log(n=3^{m})$. The lower panel shows that all five states comprising the resonance cascade become substantially populated. 
(c) A gedanken experiment where $n=9$ is excluded from the set of natural numbers. At a technical level, we create a potential with a 
spectrum $E_{n} = U_{0} \ln(n)$ ($n=1,\,2,\,3,\,\ldots,\,8,\,11,\,12,\,\ldots\,,120$) where the $n=9$ energy level is excluded. We also exclude the $n=10$ level in order to preserve the relative parity of the remaining eigenstates. Such an exclusion can be seen to impede the mobility along the energy axis. In this experiment, we apply the same ``windows of silence'' as in (b). The lower panel shows that the time evolution is close to Rabi oscillations between $|n=27\rangle$ and $|n=81\rangle$.
}
\label{f:cascade}
\end{figure*}

\begin{figure*}[t]
\begin{center}
\includegraphics[width=1.\textwidth]{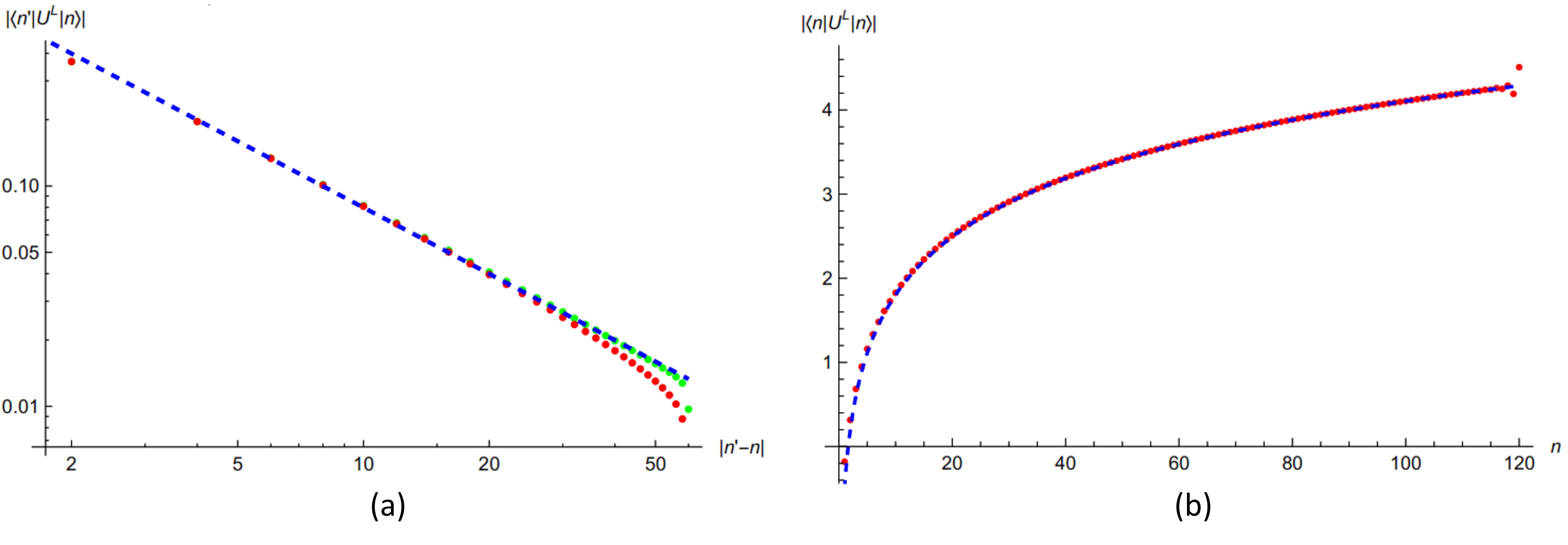}
\end{center}
\vspace{-0.7cm}
\caption{
Semiclassical approximation for the matrix elements of the perturbation (represented, in the case of a parametric excitation, by the 
potential energy) compared to the numerical results. 
(a) Absolute value of the off-diagonal matrix elements of the potential energy, $\langle n' | U^{\text{L}} | n=60 \rangle$, as a function of the absolute value of the index difference, $n'-n$. Only even values of $n'-n$ are represented, since the odd values vanish due to the parity selection rule. The green dots represent the numerical results for $n'>n$, while the red dots correspond to $n'<n$. The $\ln(n)$-spectrum potential used had $N_{\text{b}}=120$ bound states. The blue dashed line is the semiclassical prediction 
$|\langle n' | U^{\text{L}} | n \rangle| \approx A\, U_{0}/|n'-n|$ (see Eq. \eqref{off-diagonal_semiclassics}) with the value of the fit parameter $A=0.80$ representing the best fit for the data with $4 \le |n'-n| \le 30$. (b) Numerically computed diagonal matrix elements of the potential energy, $|\langle n | U^{\text{L}} | n \rangle|$ (red dots), in comparison with the virial formula \eqref{diagonal_semiclassics} (blue, dashed line); no fit parameters were used. 
}
\label{f:matrix_elements}
\end{figure*}

\subsection{Protocol}
Let us choose a natural number of interest, $\tilde{n}$ and its initial power, $m_{0}$. Our goal is to propose an effect that would be predicated on the presence of all powers of  $\tilde{n}$ (up to a practical upper limit) in the set of natural numbers.  
To this end, consider a single atom (or an ensemble of noninteracting atoms) in an excited state $|n=\tilde{n}^{m_{0}}\rangle$, with $\tilde{n}^{m_{0}} < N_{\text{b}}$. Next, we apply a parametric perturbation
\begin{align}
V(x,\,t) = \beta \cos(\Omega t) \, U^{\text{L}}(x)
\,\,,
\label{log_n_perturbation}
\end{align}
where 
\begin{align}
&
\Omega = \frac{1}{\hbar}U_{0} \ln(\tilde{n})
\,\,,
\label{log_n_perturbation_frequency}
\end{align}
with $\beta \ll 1$. The full Hamiltonian becomes
\[
\hat{H}^{\text{L}} = \frac{\hat{p}^2}{2m} + \left(1+\beta\cos(\Omega t)\right) U^{\text{L}}(x)
\,\,.
\] 

This perturbation will induce a resonance cascade, i.e., a cascade of resonant transitions between states,
\begin{widetext}
\begin{align}
|n=1\rangle \stackrel{\Omega}{\leftrightarrow} |n=\tilde{n}\rangle \stackrel{\Omega}{\leftrightarrow} |n=\tilde{n}^2\rangle\stackrel{\Omega}{\leftrightarrow} \ldots \stackrel{\Omega}{\leftrightarrow} |n=\tilde{n}^{m}\rangle \stackrel{\Omega}{\leftrightarrow} \ldots 
\stackrel{\Omega}{\leftrightarrow} |n=\tilde{n}^{m_{0}}\rangle \stackrel{\Omega}{\leftrightarrow} |n=\tilde{n}^{m_{0}+1}\rangle
\stackrel{\Omega}{\leftrightarrow}\ldots,
\,
\label{cascade}
\end{align}
\end{widetext}
which will cause the population to spread over energy levels in a resonant manner. An ability of our system to reach the $|n=1\rangle$ state would serve 
as a ``proof'' that each of the numbers $1,\,\tilde{n},\,\tilde{n}^2,\,\ldots,\,\tilde{n}^{m_{0}}$ belongs to the set of naturals. 

Our first numerical experiment features the following set of parameters: the number of bound states $N_\mathrm{b}$ was set to $120$, and we choose $\tilde{n} = 3$ and $m_0 = 3$ such that the initial state is $|n = 27\rangle$. The perturbation amplitude is $\epsilon = 0.3$. Fig.~ \ref{f:cascade}(a) shows our results. While the members of the cascade, cf. Eq.~\eqref{cascade}, clearly dominated the population 
(see lower panel), the population was localized in the vicinity of $m=3$ and $m=2$. 

In order to understand the origin of this localization, we analyze the behavior of the off-diagonal matrix elements of our perturbation, 
$\langle n' | U^{\text{L}} | n \rangle|$, as shown in Fig.~\ref{f:matrix_elements}(a). We find an inverse-linear dependence on the magnitude of the quantum 
number difference, $|n'-n|$. This assertion is well supported by the large-$|n'-n|$ asymptotics of the semiclassical approximation for the matrix elements, as shown by Eq.~\eqref{off-diagonal_semiclassics_no_fit} in Appendix \ref{ss:log_n_off-diagonal}. We obtain
\begin{align*}
\langle n' | U^{\text{L}} | n \rangle \stackrel{|n'-n|\gg 1}{\approx} 
\left\{
\begin{array}{lcc}
(-1)^{\frac{n'-n}{2}-1} \frac{U_{0}}{|n'-n|} &\text{for}& n'-n \,\text{even,}
\\
0&\text{for}& n'-n \,\text{odd.}
\end{array}
\right.
\end{align*}
Numerically, we find a dependence \[|\langle n' | U^{\text{L}} | n \rangle| \stackrel{|n'-n|\gg 1}{\approx} 
A \,\frac{U_{0}}{|n'-n|} \,\,,\] with $A=0.80$ for $n=60$, cf. Eq.~\eqref{off-diagonal_semiclassics}, while our analytic result 
\eqref{off-diagonal_semiclassics_no_fit} predicts $A=1$. This shows a good agreement between the semiclassical approximation and the numerical calculation for the matrix elements.

In terms of the cascade states, Eq.~\eqref{cascade}, an inverse linear dependence of the transition matrix elements on $|n'-n|$ leads to an exponential decay of the hopping constant within the cascade,
\[
|\langle n'=\tilde{n}^{m+1} | U^{\text{L}} | n= \tilde{n}^{m} \rangle| \sim U_{0} e^{-\gamma m},
\,\,
\] 
with $\gamma = \ln(\tilde{n})$. Under a rotating wave approximation, the system becomes a tight-binding model with the hopping coefficients $J_{m}$ (between the $m$-th and $m+1$-st sites, or states, in our model) that exponentially decay as a function of $m$: 
\begin{align}
\hat{H}_{\text{EL}} =  - J_{0}\sum_{m=-\infty}^{+\infty} e^{-\gamma m}\,\left( |m+1\rangle\langle m| + |m\rangle\langle m+1|   \right).
\label{H_EL_Bis}
\end{align}
The coefficient $J_{0}$ in this exponential lattice model is then given by 
\[
J_{0} = \frac{1}{2}U_{0} \frac{A}{(\tilde{n}-1)}\,.
\] 
In Appendix~\ref{s:exponential_lattice} we study the resulting lattice system in detail. Indeed, we find a localization for any eigenstate energy 
$E$. While a localization in the ``classically forbidden'' region of high $m$ where $|E| \gg J_{m}$ was expected, a localization in the direction of lower $m$, where $|E| \ll J_{m}$ requires an explanation. We find that in this region, the eigenstate wavefunction becomes a member of the kernel of the Hamiltonian, i.e. a state such that the action of the Hamiltonian on it is null (see App. \ref{ss:dark_states}). The structure of the such a so-called dark state wavefunction~\cite{harris1993_552,arimondo1996_257} is as follows (cf. Fig.~\ref{f:exponential_lattice__eigenstates}). The population of the odd sites is zero, as the amplitudes of the even sites form a coherent superposition to ensure that there is no population of the odd sites when the Hamiltonian in Eq.~\eqref{H_EL_Bis} acts on such a superposition. On the other hand, since the even sites have unpopulated neighbors, the action of the Hamiltonian Eq.~\eqref{H_EL_Bis} on such a state is zero.~\footnote{Interestingly, we find that any boundary condition at low $m$ would instantly favor the sites of one parity over another; in the convention used, it is indeed always the odd sites that remain unpopulated in the low $m$ tail of the eigenstates.} Finally, in Appendix \ref{ss:dark_states}, we show that the coherence of a dark state leads to an exponential decay of the population in space. 

To promote mobility along the energy axis, we consider the following amendment to our protocol. Whenever localization occurs, we stop the excitation for a quarter of its period (i.e.\ for a time interval $\Delta t = \frac{\pi}{2\Omega}$) and then restart the excitation, maintaining the phase of the drive as it was in the beginning of this so-called {\it window of silence}. As a result, the free evolution during $\Delta t$ flips the relative sign of the neighboring even sites, thus breaking the coherence of the dark state tail.  Fig.~\ref{f:cascade}(b) shows that this is indeed a valid strategy. With only two windows of silence we are able to transport the population towards the ground state while traveling exclusively through the resonance ladder.  Thus this experiment demonstrates, via simulation, that all of the powers of $3$ lower than or equal to $3^3 = 27$ are present in the set of natural numbers. 

Finally, we simulate a fictitious universe where $3^2=9$ is not a natural number. In Fig.~\ref{f:cascade}(c), we consider a potential with a spectrum similar to the simulation shown in Fig.~\ref{f:cascade}(b)  but with the $n=9$ level excluded, along with the $n=10$ level to preserve the relative parity of the remaining eigenstates. Using the same driving protocol as in Fig.~\ref{f:cascade}(b), we note that the downward energy transfer stops at the missing levels, affirming the ability of our procedure to notice holes in energy-equidistant ladders of the quantum eigenstates~\footnote{Additionally, we verify that the observed drop in the cascade population is predominantly due to excitation to the continuum: the bound states outside the cascade remain unpopulated, thus verifying the resonant model.}. 

In an experimental implementation, the dynamics shown in Fig.~\ref{f:cascade}  will be observable in less than 10~s, i.e., on a similar timescale to the dynamics presented in Section~\ref{parametricresonance}, even though the time axis in Fig.~\ref{f:cascade} extends to larger values compared to the time axis in the figures of Section~\ref{parametricresonance}. The reason is as follows: the value of the energy scale $U_0$ is larger for the logarithmic potential than for the prime number potential. This is due to the cusp-like shape of the logarithmic potential, which is narrower than the flatter profile of the prime number potential, leading to a larger spacing between the energy levels of the logarithmic potential. This larger energy scale thus gives a smaller time unit $\hbar/U_0$. Hence, the physical time span of  Fig.~\ref{f:cascade} ends up being comparable to those of the protocols in Section~\ref{parametricresonance}.

\subsection{Discussion}

For the simulations presented here, we used a parametric perturbation. We have previously considered power-law perturbations, such as the quadratic perturbation used in Section \ref{parametricresonance}, but we found that for power-law drives, the transition matrix elements change more rapidly over the cascade compared to the $|n'-n|^{-1}$ dependence for the case of the parametric perturbation. Hence power-law perturbations appear to be less favorable for cascades. However other choices of perturbation could lead to a better behavior of the matrix elements, with the ideal perturbation being one that gives approximately constant matrix elements over the cascade. This will be an area of further investigation in future work.

Secondly, in the current scheme we position the windows of silence in an {\it ad hoc} manner. Finding the optimal number and timing of the windows required numerical experimentation. In the future, we would like to find a recipe that will be able to suggest the optimal sequence of the windows of silence without a need for numerical or empirical experiments. 

\section{Future experiments with resonance cascades \label{futureideas}}

Building on the results of the previous section, here we propose that resonance cascades can be used to study other statements from number theory. In particular, we address the  Diophantus-Brahmagupta-Fibonacci identity and the Goldbach conjecture, cases (II) and (III) of Table~\ref{tab:number-theory}.

\subsection{Resonance cascade predicated on the validity of the Diophantus-Brahmagupta-Fibonacci identity}
The set of natural numbers is not the only set of integers that is closed under multiplication. As follows from the  Diophantus-Brahmagupta-Fibonacci identity
\cite{dudley_Number_Theory1970} (Lemma 1, p.\ 142), the set of integers that are equal to the sums of two squares (see Eq.~\eqref{sum_of_two_squares}) also has this property~
\footnote{%
To the contrary, the sums of three squares are not closed under multiplication, e.g.\ $11\times 373 = (1^2+1^2+3^2)\times(2^2+12^2+15^2)=4103 \neq a^2+b^2+c^2$. 
Also, according to Lagrange's four-square theorem \cite{dudley_Number_Theory1970} (p. 146), the set of integers that are equal to the sums of four squares coincides with the set of all natural numbers.%
}~\footnote{A relevant fact: if $a$ is an integer that is equal to the sum of two squares, then the prime decomposition of $a$ cannot contain a factor $p^{k}$ such that $p$ is a prime of the form $4 \times n+3$, where $n$ is an integer and $k$ is odd.}. An example of a potential whose spectrum is proportional to the sequence in Eq.~\eqref{sum_of_two_squares} is shown in Fig.~\ref{f:potentials_CUMULATIVE}(b). 

Similarly to the case of a $\ln(n)$-spectrum potential shown in Fig.\ref{f:potentials_CUMULATIVE}(a), the $\ln(s)$-spectrum potential, 
\begin{align}
\begin{split}
&
E^{\text{L2}}_{n}  = U_{0} \ln s_{n} 
\\
&
s_{1},\,  s_{2},\, s_{3},\, s_{4},\, \ldots = 
\\
&
= (0^2+1^2),\, (1^2+1^2),\, (0^2+2^2),\, (1^2+2^2),\,\ldots
\\
&
n=1,\, 2,\, 3,\,\ldots,\, N_{\text{b}}.
\end{split}
\label{sum_of_two_squares}
\end{align}
where the integers $s$ are each the sum of two squares, will support unbounded cascades of resonant transitions for a single particle in the $\ln(s)$-spectrum potential.  Such cascades require that the frequency of the perturbation must have the form
\begin{align*}
\Omega = \frac{1}{\hbar} U_{0} \ln \tilde{s}
\,\,,
\end{align*}
where $\tilde{s}$ is itself a sum of two squares. Note that for technical reasons, we exclude $0^2+0^2=0$ from the set in Eq. \eqref{sum_of_two_squares}.

\subsection{Resonance cascade predicated on the validity of the Goldbach conjecture}
Unlike the $\ln(n)$-spectrum potential considered above, a one-body potential whose spectrum is given by the prime numbers has already been demonstrated theoretically and experimentally~\cite{cassettari2022_220203446}. We examined state manipulation and control in this potential in Sec.~\ref{parametricresonance}. As a slight modification, in this section, we will exclude $2$ from the set of primes, given that $2$ only enters in the Goldbach decomposition of $4$, and not in the Goldbach decomposition of higher even numbers. Then, the potential will have the following one-body spectrum:  
\begin{align}
\begin{split}
&
E^{\text{P}}_{n}  = U_{0} p_{n} 
\\
&
p_{1},\,  p_{2},\, p_{3},\, p_{4},\, \ldots = 3,\, 5,\, 7,\,11,\,\ldots
\\
&
n=1,\, 2,\, 3,\,\ldots,\, N_{\text{b}}.
\end{split}
\,\,,
\label{primes}
\end{align}
\begin{figure}[!t]
\begin{center}
\includegraphics[width=\columnwidth]{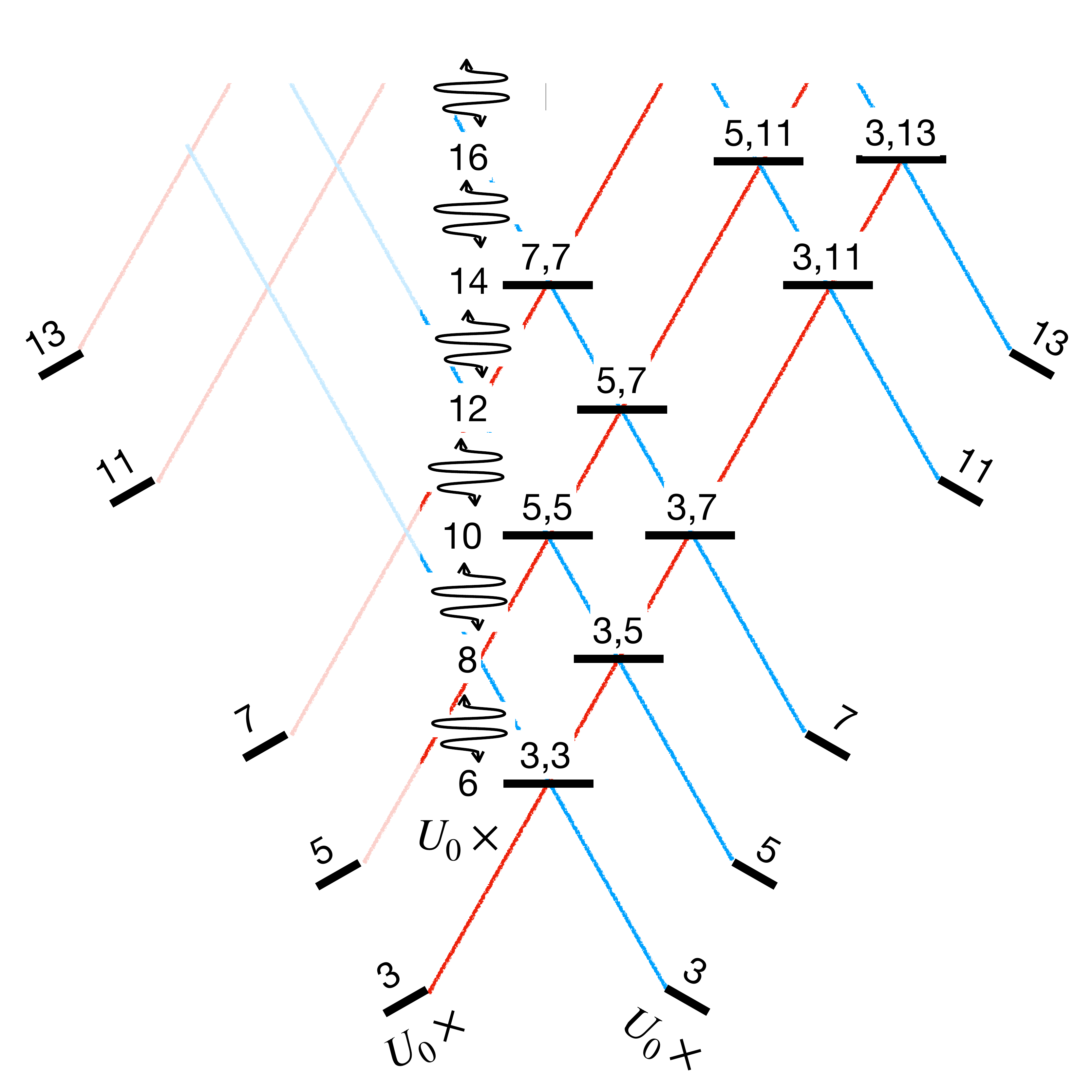}
\end{center}
\caption{Energy spectrum of two weakly interacting bosons in a \textit{p}-spectrum potential. The energy levels, $E = U_{0} \times 4,\,6,\,8,\,\ldots$, are formed by sums of two prime numbers, $p_{1}$ and $p_{2} \ge p_{1}$. In the text, we suggest applying a weak periodic one-body resonant perturbation of frequency $2U_{0}/\hbar$. Up to the level $E= 38U_{0}$, the one-body transitions alone ensure the resonant upward mobility of the system, thanks to a proliferation of twin prime pairs. The gap between $E=38U_{0}$ and $E'=40U_{0}$, which is inaccessible by one-body processes, can be overcome by weakly mixing the original eigenstates using a two-body interaction. Finally, a complete upward mobility of the level population is predicated on the validity of Goldbach's conjecture. 
}
\label{f:goldbach}
\end{figure}

Consider two bosonic atoms in this \textit{p}-spectrum trap (see Fig.~\ref{f:goldbach}). For the moment, let us set interactions to zero. Let us apply a weak {\it one-body} periodic perturbation $\hat{V} \cos(\Omega t)$ of frequency $\Omega = 2U_{0}/\hbar$. This perturbation changes the energy of one of the atoms by $2U_{0}$, while leaving the energy of the other atom unchanged. Let us start from the ground state, $|3,\,3\rangle$, and explore the possibility of traveling infinitely high in energy, via resonant transitions {\it only}. 

Consider two energy-consecutive two-body eigenstates $|p_{1},\,p_{2}\rangle$ and $ |p'_{1},\,p'_{2}\rangle$, with
\begin{align*}
&
w=p_{1}+p_{2},
\\
&
w'=p'_{1}+p'_{2},
\\
&
p_{2} \ge p_{1},
\\
&
p'_{2} \ge p'_{1},\
\\
&
\text{and}
\\
&
w'=w+2,
\end{align*}
where $w$ is an even integer. The energies of these states are are given by $E = \hbar U_{0} w$ and $E' = \hbar U_{0} w' = E + 2U_{0}$, respectively.
If at least one of the primes $p_{1},\,p_{2}$ participating the Goldbach decomposition of $w$ is a lower member of a twin prime pair (a ``lower twin'' from now on), where twin primes are a pair of primes separated by 2, 
then $|p_{1},\,p_{2}\rangle$ and $ |p'_{1},\,p'_{2}\rangle$  will be resonantly coupled with a nonzero matrix element. That is, the one-body nature of the perturbation imposes the following selection rule for the transition $w \to w'=w+2$:

\begin{align}
\begin{split}
\\
&
p'_{1} = p_{1}+2 \text{ and } p'_{2} = p_{2}
\\
& \text{or}
\\
&
p'_{2} = p_{1}+2 \text{ and } p'_{1} = p_{2}
\\
& \text{or}
\\&
p'_{1} = p_{2}+2 \text{ and } p'_{2} = p_{1}
\\
& \text{or}
\\&
p'_{2} = p_{2}+2 \text{ and } p'_{1} = p_{1}.
\end{split}
\label{one_body_selection_rules}
\end{align}
For the allowed transitions, the transition matrix element can be estimated as
\begin{align*}
\langle p'_{1},\,p'_{2}| \hat{V} \otimes \hat{I} + \hat{I} \otimes \hat{V} |p_{1},\,p_{2}\rangle \sim \langle p'| \hat{V} |p\rangle
\,\,,
\end{align*}
where $\langle p'| \hat{V} |p\rangle$ is a typical one-body matrix element of $\hat{V}$.

The first one-body-forbidden transition occurs at 
\begin{align*}
&
w=38=7+31=19+19
\\
&
w'=40=3+37=11+29=17+23
\,\,,
\end{align*}
and it is not known yet whether the number of such transitions is infinite or finite. 
To break the selection rules in Eq.~\eqref{one_body_selection_rules}, one will need to add a second potential, e.g., a weak two-body interaction potential $\hat{V}^{(2)}$. In this case, the previously forbidden transition $|p_{1},\,p_{2}\rangle  \to |p'_{1},\,p'_{2}\rangle$ becomes weakly allowed. To first order in perturbation theory with respect to $\hat{V}^{(2)}$, the transition matrix element is estimated by
\begin{widetext}
\begin{align}
\langle p'_{1},\,p'_{2}| \hat{V} \otimes \hat{I} + \hat{I} \otimes \hat{V} |p_{1},\,p_{2}\rangle \sim \frac{\langle p'| \hat{V} |p\rangle \langle p'_{I},\,p'_{II}| \hat{V}^{(2)} |p_{I},\,p_{II}\rangle }{U_{0}}
\,\,,
\label{two-body_interaction_transition}
\end{align}
\end{widetext}
where $\langle p'_{I},\,p'_{II}| \hat{V}^{(2)} |p_{I},\,p_{II}\rangle$ is a typical two-body matrix element of $\hat{V}^{(2)}$. 

In Appendix \ref{s:twin_primes}, we prove that the number of two-body energy levels that require two-body interactions to be resonantly coupled with the next level above is infinite (equivalently, that there exists an infinite number of even numbers whose Goldbach decomposition does not involve lower twin primes). Finally, it may happen that for a given two-body energy level $E$, the required level above, with energy $E + 2U_{0}$, has no eigenstates. In the language of number theory, this would mean that the even number $w'$ can not be represented as a sum of two primes. However, that would contradict Goldbach's conjecture, which states that every even number greater than $2$ is the sum of two prime numbers. Note that here, as noted previously, we consider an equivalent formulation where every even number greater than $4$ is a sum of two primes, each greater than $2$. If we assume that Goldbach's conjecture is valid, we can expect that starting from the ground state of our system, we can reach any of its excited states via resonant transitions alone.

According to the preliminary study~\cite{marchukov2024}, the atomic population can be exclusively contained within a resonantly coupled subset of energy levels if the typical values of the off-diagonal matrix elements of the perturbation are of the order of $|\langle p'| \hat{V} |p\rangle| \sim 10^{-4} \times U_{0}$ or lower, with $U_{0}$ representing the typical level spacing. Such a stringent requirement reflects the fact that experimentally accessible perturbations
contribute to both the desired transition matrix elements between the cascade states and the parasitic off-resonant  shifts of the unperturbed energy levels. The resulting uppper bound on the value of $|\langle p'| \hat{V} |p\rangle|$  is a result of a compromise between the necessity to ensure the dominance of the resonant processes and the desire to minimize propagation time. Note that the corresponding parasitic level shift is a second order perturbation theory effect, and as such, should be of the order of 
$|\langle p'| \hat{V} |p\rangle|^2/U_{0} \sim 10^{-8} \times U_{0}$ or lower.

On the other hand, the two-body interaction, $\langle p'_{I},\,p'_{II}| \hat{V}^{(2)} |p_{I},\,p_{II}\rangle$, is a time-independent perturbation, and it will affect the level shifts in the first order of the perturbation theory. Assuming that
the resulting undesirable level shifts 
$|\langle p'_{I},\,p'_{II}| \hat{V}^{(2)} |p_{I},\,p_{II}\rangle$
are also of the order of $10^{-8} \times U_{0}$, the resulting allowed values of the transition matrix element \eqref{two-body_interaction_transition} become impractically 
small, of the order of $10^{-12} \times U_{0}$ or less. One way to overcome this difficulty will be to offset the two-body-interaction-induced level shifts using a priori corrections to the unperturbed energy levels. In this case, the matrix element 
\eqref{two-body_interaction_transition} can, conceivably, reach a value comparable to one of the couplings between the regular two-body energy levels.

\section{Conclusions}
In this paper we have examined two types of potentials, prime number potentials and logarithmic potentials, that are relevant for applications of quantum physics to number theory. First, we studied the case in which the system behaves as a two-level system due to the unequal spacing between the energy levels. We have shown that, starting from the ground state, we can prepare an individual excited state with high fidelity. Moreover, transfer time can be reduced with quantum control techniques, which allows for faster state preparation for the experiments we consider in the rest of this work. 

Next, we outlined a strategy for studying problems of number theory using quantum systems. The strategy consists of reducing a number-theoretic statement to an assertion of the existence of an uninterrupted arithmetic sub-sequence in a sequence of numbers, and subsequently finding a quantum system with an energy spectrum proportional to this sub-sequence. The existence of an arithmetic sub-series is illustrated by a cascade of resonant transitions under a periodic perturbation. 

In this spirit, we have studied the example of the closeness of the set of natural numbers under multiplication in great detail. In doing so, we identified a difficulty that may occur in several instances of our class of protocols: the hopping constants in the resulting resonance cascades exhibit a sharp dependence on the site number. As a result, the eigenstates become localized inhibiting mobility along the energy axis. Subsequently, we found a way to unlock the mobility, so far in the downward direction only. The recipe 
calls for introducing a few specifically designed short windows of silence during which the perturbation is turned off. Hence this protocol requires, as the starting point, the preparation of an excited state of the potential. Such preparation could be performed using quantum optimal control methods like those presented in Sec.~\ref{ss:QC}. Due to the logarithmic scaling of the energy levels (which leads to continuum-like behavior at higher levels but relatively large energy spacings between, e.g., the ground and first excited state), we found that gradient-based methods like GRAPE~\cite{Khaneja2005} tend to fail. On the other hand, gradient-free methods like genetic algorithms or methods relying on the constrained optimization within a given basis, like the the CRAB~\cite{CRAB} or GROUP~\cite{GROUP} methods will likely succeed, and both methods have been shown to work for similar atomic state preparation problems, as for instance in Refs.~\cite{vanFrank2016,Weidner2018,Omran2019-hy,Muller2022,Parajuli2023}.

In the future, we plan to work on further improving the quantum control schemes to transfer particles from low-energy to high-energy eigenstates. Effect of temperature, interactions and transverse degrees of freedom should be also studied in detail to optimize the efficiency of number-theory-inspired devices. We also plan to design schemes in which one starts from the ground state of the trap and ignites a resonance cascade propagating upward in energy. In this regard, preliminary results have been shown in related work~\cite{marchukov2024}. We envisage that this will make a future experimental implementation easier. Experimentally, it is easier to prepare an atom in the ground state than in a well-defined excited state. Moreover, an upward-oriented resonance cascade may be readily detectable by counting atoms that leave the trap completely, as they reach the threshold between the bound and the continuous spectra. Hence atom loss would be a clear experimental signature of an uninterrupted cascade. 

\begin{acknowledgments}
HK and CW wish to acknowledge the computational facilities of the Advanced Computing Research Centre, University of Bristol - http://www.bris.ac.uk/acrc/. HK was funded by the EPSRC CDT in Quantum Engineering (EP/SO23607/1). CW acknowledges funding from EPSRC under grant number EP/Y004728/1. 
MO was supported by the NSF Grant No.~PHY-2309271. 
OVM acknowledges the support by the DLR German Aerospace Center with funds provided by the Federal Ministry for Economic Affairs and Energy (BMWi) under Grant No. 50WM1957 and No. 50WM2250E. 
The authors would like to thank the Institut Henri Poincar\'{e} (UAR 839 CNRS-Sorbonne Université) and the LabEx CARMIN (ANR-10-LABX-59-01) for their support.
GM and AT thank the kind hospitality and support from the program
``Out-of-equilibrium Dynamics and Quantum Information of Many-body Systems with Long-range Interactions'' at KITP (Santa Barbara). GM acknowledges the grants PNRR
MUR Project PE0000023-NQSTI and PRO3 Quantum Pathfinder.

A significant portion of this work was produced during the thematic trimester on ``Quantum Many-Body Systems Out-of-Equilibrium'', at the  Institut Henri Poincaré (Paris): AT, GM, and MO are grateful to the organizers of the trimester, Rosario Fazio, Thierry Giamarchi, Anna Minguzzi, and Patrizia Vignolo, for the opportunity to be a part of the program.
\end{acknowledgments}

\appendix

\section{The $\ln(n)$-spectrum potential: matrix elements of the perturbation}
In this article, we are concerned with the design of an experimental protocol that can make apparent the existence of equidistant energy level ladders. Our method of choice is the study of the mobility of an initially prepared quantum state along the energy axis under a resonant perturbation. The design process involves analyzing simple numerical and analytic models of the process; these appendices are devoted to that goal. 

In the main text, we suggest using a parametric drive for inducing a resonance cascade. In particular, given the potential $U^{\text{L}}(x)$ that produces a logarithmic spectrum, Eq. \eqref{log_n},  
\begin{align*}
&
\hat{H} = \frac{\hat{p}^2}{2m} + U^{\text{L}}(x)
\\
&
\hat{H} |n\rangle = E_{n} |n\rangle
\\
&
E_{n} = U_{0} \ln(n)
\\
&
n=1,\,2,\,3,\,\ldots
\,\,,
\end{align*}
we apply a perturbation
\begin{align*}
V(x,\,t) = \beta \cos(\Omega t) U^{\text{L}}(x)
\,\,,
\end{align*}
where 
\begin{align*}
&
\Omega = U_{0} \ln(\tilde{n})
\,\,,
\end{align*}
and $\tilde{n}$ is a natural number. Note also that in the main text, we supplement the $\cos(\Omega t)$ time dependence of the perturbation with a few quarter-period-long windows of silence, as a method to enhance the mobility along the energy axis. 

In the view of the above, estimates for the matrix elements of the potential $\langle n' | U^{\text{L}} | n \rangle$ will become needed. These estimates are the goal of this appendix. 

\subsection{A logarithmic approximation to the  $\ln(n)$-spectrum potential \label{ss:logarithmic_potential}}
In classical mechanics, the solution to the problem of restoring a trapping potential from its frequency as a function of energy 
is well-known \cite{landau_classical___potential_from_period}:
\begin{align*}
x(U) = \frac{1}{\sqrt{2m}} \int_{U(0)}^{U} \frac{dE}{\omega(E)\sqrt{U-E}}
\,\,,
\end{align*}
where $\omega(E)$ is the oscillation frequency as a function of energy, $x(U)$ is the inverse of the potential function 
$U(x)$ for $x\ge 0$, $U(0)$ is the value of the potential at the origin, and the potential $U(x)$ is assumed to be an even function of $x$ such that $U(-x) = U(x)$. On the other hand, given a quantum energy spectrum, $E_{n}$, the classical frequency can readily be extracted as 
\begin{align}
\omega = \frac{1}{\hbar} \frac{d E_{n}}{dn} \,\,.
\label{omega_through_E}
\end{align}
In our case, the spectrum $E_{n}$ is given by Eq.~\eqref{log_n}. Then the classical frequency as a function of energy becomes 
\begin{align}
\omega = \frac{1}{\hbar} U_{0} e^{-\frac{E}{U_{0}}} \,\,.
\label{classical_omega}
\end{align}
Now, assuming $U_{0} = -\infty$, we get Eq.~\eqref{log_n_potential_classical}. Fig.~\ref{f:potentials_CUMULATIVE}(a) shows a good agreement between the quantum $\ln(n)$-spectrum potential and its classical counterpart everywhere except for the substantially quantum region 
$|x|\lesssim a$ (whose size is comparable to the spatial extent of the quantum ground state).

\subsection{Classical motion in a logarithmic potential}
Classical equations of motion generated by the Hamiltonian
\begin{align*}
&
H^{\text{L}}_{\text{cl.}} = \frac{p^2}{2m} + U^{\text{L}}_{\text{cl.}}(x) 
\,\,,
\end{align*}
with $U^{\text{L}}_{\text{cl.}}(x)$ given by Eq.~\eqref{log_n_potential_classical}, can be readily solved, albeit implicitly: 
\begin{align}
t = \pm T(E) 
\left\{
\begin{array}{lcc}
+\frac{1}{4} \erf\left(\sqrt{\ln(\frac{b(E)}{x})}\right) & \text{for} & x>0
\\
-\frac{1}{4} \erf\left(\sqrt{\ln(\frac{-b(E)}{x})}\right) + \frac{1}{2} & \text{for} & x<0 
\end{array}
\right.
\,\,,
\label{log_n_potential_classical__trajectory}
\end{align}
where $T(E) = (2\pi)/\omega(E)$ is the classical period, $\omega(E)$ is the classical frequency given by Eq.~\eqref{classical_omega}, and 
\[b(E) = \sqrt{\frac{\pi}{2}} a e^{\frac{E}{U_{0}}}\] (see Fig.~\ref{f:x_of_t}). 
\begin{figure}[!h]
\begin{center}
\includegraphics[width=.4\textwidth]{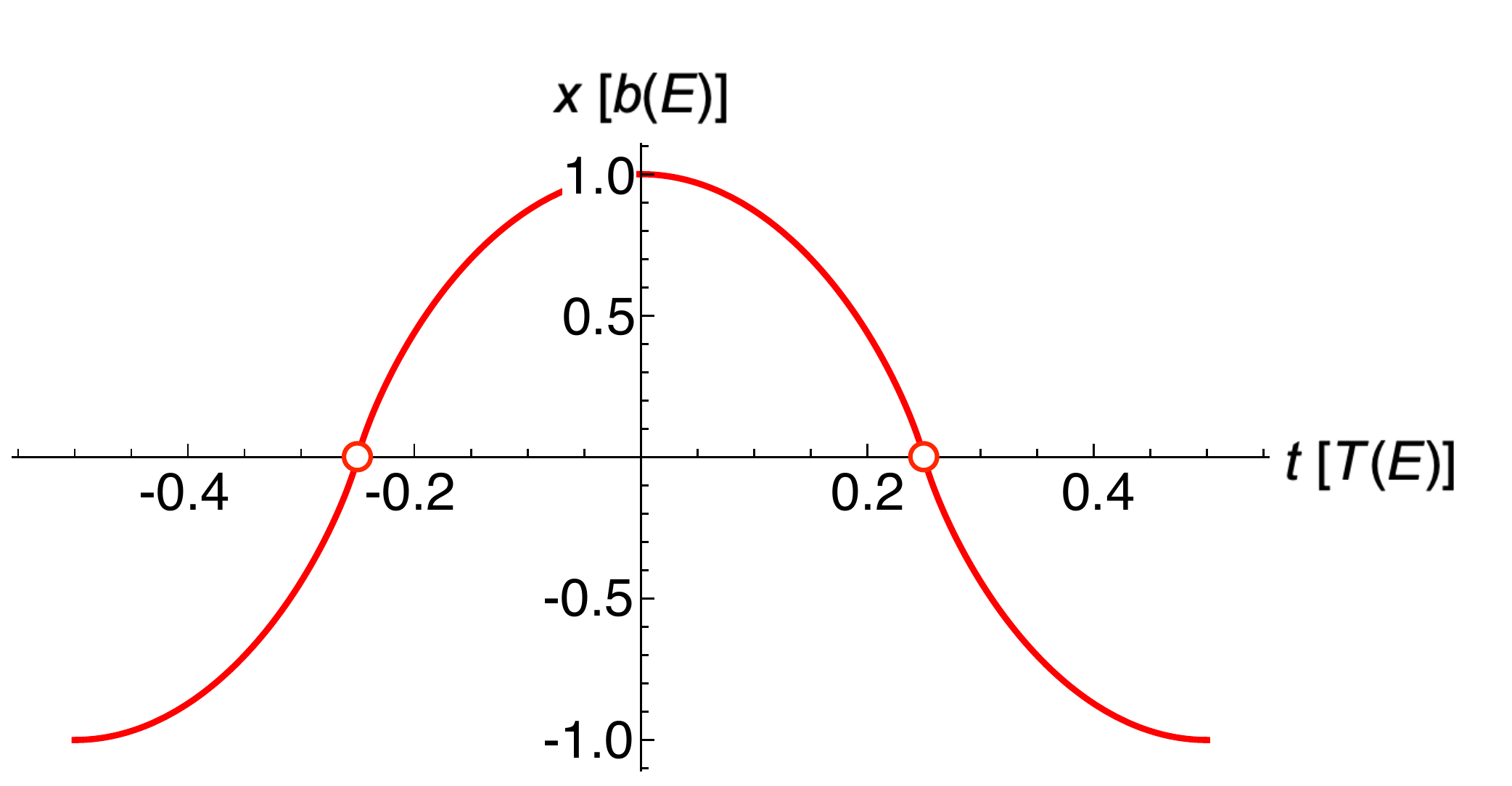}
\end{center}
\caption{Classical trajectory in a logarithmic potential.  
The time dependence of the $x$-coordinate is given by a universal function, Eq.~\eqref{log_n_potential_classical__trajectory}, with the time $t$ and $x$-coordinate properly rescaled in terms of the period $T(E)$ and the amplitude $b(E)$. Note that the $x(t)$ curve experiences singularities at $t=\pm \frac{T}{4}$ where the particle velocity diverges logarithmically (open circles), but this is not readily visible in the plot. These singularities uniquely determine the high-frequency behavior of the Fourier components
of the classical time dependence of the perturbation and, hence, its quantum matrix elements between the energy-distant unperturbed eigenstates.} 
\label{f:x_of_t}
\end{figure}
At $t=\pm \frac{T}{4}$ particle's trajectory experiences logarithmic singularities
\begin{align}
t \mp \frac{T}{4} \approx \mp \frac{T}{4\sqrt{\pi}} \frac{x}{b \sqrt{\ln(\frac{b}{|x|})}} 
\,\,.
\label{log_n_potential_singularities}
\end{align}
In the vicinity of these singularities, the implicit relations in Eq.~\eqref{log_n_potential_singularities} can be approximately inverted, using sequential iterations:
\begin{align*}
&
x_{n+1} = \mp 4\sqrt{\pi} b \left(\frac{t\mp \frac{T}{4} }{T} \right) \sqrt{\ln(\frac{b}{|x_{n}|})}, 
\\
&
x_{0} = \mp 4\sqrt{\pi} \,b\, \left(\frac{t\mp \frac{T}{4} }{T} \right).
\end{align*}
A comparison between the $x_{1}$ and $x_{2}$ iterations shows that 
\begin{align}
x(t) \approx x_{1}(t) = \mp 4\sqrt{\pi} b \left(\frac{t\mp \frac{T}{4} }{T} \right) \sqrt{\ln(\frac{T}{4\sqrt{\pi}(t\mp \frac{T}{4} ) })}
\,\,,
\label{log_n_potential_singularities_2}
\end{align}
constitutes an accurate approximation to the true trajectory if
\[ \sqrt{\ln(4\sqrt{\pi} (t\mp \frac{T}{4} ))} \ll \frac{T}{4\sqrt{\pi} (t\mp \frac{T}{4} )} \,\,;  \]
this condition is satisfied for 
\[
t\mp \frac{T}{4} \ll 0.1T
\,\,.
\]

In addition, let us observe the following. First, the asymptotic expansion of the $\erf$-function, 
\[\erf(y) \stackrel{y\gg 1}{=} 1- e^{-y^2}\left\{ -\frac{1}{\sqrt{\pi}y} + \frac{1}{2\sqrt{\pi}y^3} + \mathcal{O}(\frac{1}{y^5}) \right\} \,\,,  \] 
shows that the approximation in Eq.~\eqref{log_n_potential_singularities} is valid for 
\[
\ln(\frac{b}{|x|}) \gtrsim \frac{1}{2}
\,\,.
\]
On the other hand, the logarithmic singularity in the classical potential given by Eq.~\eqref{log_n_potential_classical} is regularized in the quantum case 
(Fig.~\ref{f:potentials_CUMULATIVE}(a)), at $|x|\sim a$. This cut-off introduces the following bound:
\[
\ln(\frac{b}{|x|}) \lesssim n
\,\,,
\] 
where $n$ is the typical eigenstate index of interest. In our case, $3 \lesssim n \lesssim 27$. 

Observe that in Eq.~\eqref{log_n_potential_singularities}, another constant $b$, outside of the logarithm, depends on $n$ exponentially. 
Accordingly, in these estimates, we will replace the logarithm appearing in Eq.~\eqref{log_n_potential_singularities} by a constant $L$:
\begin{align}
\begin{split}
&
\ln(\frac{b}{|x|}) \to L
\\
&
\frac{1}{2} \lesssim L \lesssim 27
\end{split}
\,\,.
\label{log_fudge}
\end{align}
Interestingly, in the estimates that follow, the actual value of the constant $L$ turns out to be irrelevant. As it follows from our derivation below, the singularity in Eq.~\eqref{log_n_potential_singularities} controls the behavior of the quantum off-diagonal matrix elements of the potential energy, for large differences between the quantum numbers.   

\subsection{Semiclassical approximation for the  off-diagonal matrix elements of coordinate-dependent observables in one-dimensional 
traps: general results}
In Ref.~\cite{landau_quantum___quasiclassics}, one can find a set of results on a semi-classical approximation for off-diagonal matrix elements of quantum 
observables, both for the case of two close energies (\S 48 of Ref.~\cite{landau_quantum___quasiclassics}), and for two energies far apart (\S 51). Before we address the question of the matrix elements of the potential $U^{\text{L}}(x)$ that generates the $\ln(n)$ spectrum given in Eq.~\eqref{log_n}, we will restate the main results in the former limit (shown to be relevant to our case), focusing on the case where the observable is a function of a coordinate.   

Consider a one-dimensional quantum potential well with a potential energy $U(x)$. Let $\hat{A} = A(x)$ be the observable of interest. We will be interested in a semi-classical approximation to the matrix elements of the observable $\langle n' | \hat{A} | n \rangle$, where  $| n \rangle$, $| n' \rangle$ indicate eigenstates of the system; the corresponding eigenenergies are $E_{n}, E_{n'}$. 

In the classically allowed region, $E \ge U(x)$, the semi-classical approximation to the eigenstate wavefunction is
\begin{align}
\psi_{n}(x) = \frac{2}{\sqrt{T(E_{n}) v(x,\,E_{n})}} \cos(\frac{1}{\hbar} \int_{x_{1}(E_{n})}^{x} p(x',\,E_{n}) dx' -\frac{\pi}{4}) 
\label{general_WKB_wavefunction}
\end{align}
where $T(E) = 2\pi/\omega(F)$ is the classical oscillation period as a function of energy, $\omega(E)$ is the classical frequency, \[p(x,\,E) \equiv \sqrt{2m(E-U(x))}\] is the magnitude of the classical momentum as a function of coordinate and energy, $v(x,\,E) \equiv p(x,\,E)/m$ is the the magnitude of the classical velocity, and $x_{1}(E)$ is the left turning point of particle's trajectory. i.e.\ the smallest of the two solutions of an algebraic equation $U(x) = E$. Within the accuracy of the semiclassical approximation, the wavefunction Eq.~\eqref{general_WKB_wavefunction} is normalized to unity. More specifically, its normalization integral is unity if one replaces the $\cos^2(\ldots)$ appearing there by $\frac{1}{2}$ (see \S48 in \cite{landau_quantum}).

Next, we will introduce an energy $E$ close to each of the energies $E_{n}$ and $E_{n'}$, i.e., $E\approx E_{n},\,E_{n'}$. A concrete choice for this energy is, to the leading order of the semiclassical approximation, irrelevant. The choice  $E= (E_{n}+E_{n'})/2$ may, in some cases, improve accuracy beyond the leading order \cite{olshanii_book_back_of_envelope}; this choice also preserves the Hermitian nature of the matrix for the given observable. In this text, we will not commit to any particular convention.

For a matrix element of an observable $A(x)$, between an eigenstate $|n'\rangle$ and an eigenstate $|n\rangle$, we get:
\begin{widetext}
\begin{align}
\begin{split}
\langle n' | \hat{A} | n \rangle &=
\int_{-\infty}^{+\infty} dx A(x) \psi_{n'}(x) \psi_{n}(x)
\\
&\approx
\frac{4}{T(E)}\int_{x_{1}(E)}^{x_{2}(E)}  \frac{dx}{v(x,\,E)} A(x) 
\cos(\frac{1}{\hbar} \int_{x_{1}(E_{n'})}^{x} p(x',\,E_{n'}) dx' -\frac{\pi}{4}) 
\cos(\frac{1}{\hbar} \int_{x_{1}(E_{n})}^{x} p(x'',\,E_{n}) dx'' -\frac{\pi}{4})
\\
&\approx
\frac{2}{T(E)}\int_{x_{1}(E)}^{x_{2}(E)}  \frac{dx}{v(x,\,E)} A(x) 
\cos(\int_{x_{1}(E)}^{x} dx' (p(x,\,E_{n'})-p(x,\,E_{n})))
\\
&\approx
\frac{2}{T(E)}\int_{x_{1}(E)}^{x_{2}(E)}  \frac{dx}{v(x,\,E)} A(x) 
 \cos(\frac{1}{\hbar} \int_{x_{1}(E)}^{x} dx' (p(x',\,E_{n'})-p(x',\,E_{n})))
\\
&\approx
\frac{2}{T(E)}\int_{x_{1}(E)}^{x_{2}(E)}  \frac{dx}{v(x,\,E)} A(x) 
 \cos(\int_{x_{1}(E)}^{x} \frac{dx'}{v(x,\,E)} \omega(E)(n'-n))
 \\
&=
\frac{2}{T(E)}\int_{t_{1}}^{t_{1}+\frac{T(E)}{2}}  dt \,A(x(t)) 
 \cos((t-t_{1}) \omega(E)(n'-n))
\\
&=
\frac{1}{T(E)}\int_{t_{1}}^{t_{1}+\frac{T(E)}{2}}  dt \,A(x(t)) 
 \cos((t-t_{1}) \omega(E)(n'-n)) + \frac{1}{T(E)}\int_{t_{1}}^{t_{1}+\frac{T(E)}{2}}  dt \,A(x(t)) 
 \cos((t-t_{1}) \omega(E)(n'-n)) 
\\
&=
\frac{1}{T(E)}\int_{t_{1}}^{t_{1}+\frac{T(E)}{2}}  dt \,A(x(t)) 
 \cos((t-t_{1}) \omega(E)(n'-n))  
 +
\frac{1}{T(E)}\int_{t_{1}-\frac{T(E)}{2}}^{t_{1}}  dt \,A(x(t)) 
 \cos((t-t_{1}) \omega(E)(n'-n)). 
\end{split}
\,\,.
\label{off-diagonal_derivation}
\end{align}
\end{widetext}
Above a semiclassical regime where the potential $U(x)$ does not change appreciably on the scale of the de Broglie wavelength, 
$\lambda_{\text{d-B}} \sim \hbar/p$, is assumed; we also used the relationship in Eq.~\eqref{omega_through_E}. Here and below, $t_{1}$ is the time when the particle reaches the left turning point, i.e., where $x(t_{1}) = x_{1}$. We also used $dt = dx/v(x,\,E)$. Finally, we used the fact the functions $A(x(t))$ and $\cos((n'-n) \omega(E)(t-t_{1}) )$ are even with respect to a $(t-t_{1}) \to -(t-t_{1})$ sibstitution and periodic with a period $T(E)$. We arrive at the familiar semiclassical expression for the matrix elements of a coordinate-dependent observable $A(x)$ between two close energy levels~\cite{landau_quantum___quasiclassics} 
\begin{align}
\langle n' | \hat{A} | n \rangle  
\approx \frac{1}{T(E)}\int_{-\frac{T(E)}{2}}^{+\frac{T(E)}{2}}  d\tau \,A(x(t_{1}+\tau)) \cos(\Delta n \, \omega(E)\tau )
\label{off-diagonal_general_result_alt}
\\
\approx \frac{(-1)^{\Delta n}}{T(E)}\int_{-\frac{T(E)}{2}}^{+\frac{T(E)}{2}}  d\tau' \,A(x(t_{2}+\tau')) \cos(\Delta n \, \omega(E)\tau' )
\label{off-diagonal_general_result_alt_2}
\,\,.
\end{align}
Under this approximation, the matrix element in question becomes a cosine Fourier component of the time dependence of the classical counterpart. Here, $\Delta n=n'-n$, $E \approx E_{n},\,E_{n'}$, $t_1 (t_2)$ are the left (right) turning point times, and $\omega(E)$ is the frequency at energy $E$. 
Note that for this choice of the lower bound of the Fourier integral, the corresponding sine Fourier component vanishes.  

\subsection{Off-diagonal matrix elements of the potential energy for the $\ln(n)$-spectrum potential \label{ss:log_n_off-diagonal}}
In this subsection we are going to use the semi-classical formula Eq.~\eqref{off-diagonal_general_result_alt_2} to estimate the off-diagonal matrix elements of the potential energy $\langle n' | U^{\text{L}} | n \rangle$, in conditions where the quantum number difference $n'-n$ is large compared to unity so that the existing intuition for asymptotic expansions of Fourier integrals is applicable~\cite{bleistein_Fourier_theory_1986} but small compared to $n$. The latter constraint ensures that the formula \eqref{off-diagonal_general_result_alt_2} is still valid. Within the accuracy of the semiclassical approximation, we can replace the $\ln(n)$-spectrum potential Eq.~\eqref{log_n} with its classical counterpart Eq.~\eqref{log_n_potential_classical}. The classical trajectory is given by Eq.~\eqref{log_n_potential_classical__trajectory}. 

According to Ref.~\cite{bleistein_Fourier_theory_1986}, at large orders of a Fourier expansion, the Fourier components are dominated by the singularities in the time-dependence. In our case, we are dealing with logarithmic singularities. Indeed, substituting the approximation Eq.~\eqref{log_n_potential_singularities_2} 
to the potential Eq.~\eqref{log_n_potential_classical} and using the substitution Eq.~\eqref{log_fudge}, one gets 
\begin{align}
U^{\text{L}}_{\text{cl.}}(x(t)) \stackrel{t\approx \mp \frac{T}{4}}{\approx} U_{0} \ln(4\sqrt{2}\left(\frac{b}{a}\right) \sqrt{L} \left(\frac{|t\mp \frac{T}{4}|}{T}\right))
\,\,.
\label{potential_of_t_singularity}
\end{align}
To estimate the Fourier integral in Eq.~\eqref{off-diagonal_general_result_alt_2} at large quantum number differences ($1 \ll |n'-n| \ll n$), we extend the integral to the full axis of time, use the approximation in Eq.~\eqref{potential_of_t_singularity}, and introduce an ultraviolet cut-off, shown later to be removable. The following integral emerges:
\begin{align*}
\lim_{\lambda\to 0} \int_{-\infty}^{+\infty} d\tau \exp(i \tilde{\Omega} \tau) \exp(-\lambda |\tau|) \ln(\sigma |\tau|) 
= -\frac{\pi}{|\tilde{\Omega}|}
\,\,.
\end{align*}
The off-diagonal matrix elements of the potential can be estimated as 
\begin{align}
\langle n' | U^{\text{L}} | n \rangle \stackrel{|n'-n|\gg 1}{\approx} 
\left\{
\begin{array}{lcc}
(-1)^{\frac{n'-n}{2}-1} \frac{U_{0}}{n'-n} &\text{for}& n'-n = \text{even}
\\
0&\text{for}& n'-n = \text{odd}.
\end{array}
\right.
\label{off-diagonal_semiclassics_no_fit}
\end{align}
A comparison with the numerical results shown in Fig.~\ref{f:off-diagonal_FULL_INFO} shows that while Eq.~\eqref{off-diagonal_semiclassics_no_fit} captures the overall scaling, 
\begin{align}
|\langle n' | U^{\text{L}} | n \rangle| \stackrel{|n'-n|\gg 1}{\approx} 
A \,\frac{U_{0}}{|n'-n|}
\label{off-diagonal_semiclassics}
\,\,,
\end{align}
it  overestimates the coefficient $A$ by $25\%$, predicting that $A=1$ instead of the numerically-obtained $A=0.80$ 
(cf. Fig.~\ref{f:matrix_elements}(a)).

\begin{figure}[!h]
\begin{center}
\includegraphics[width=.5\textwidth]{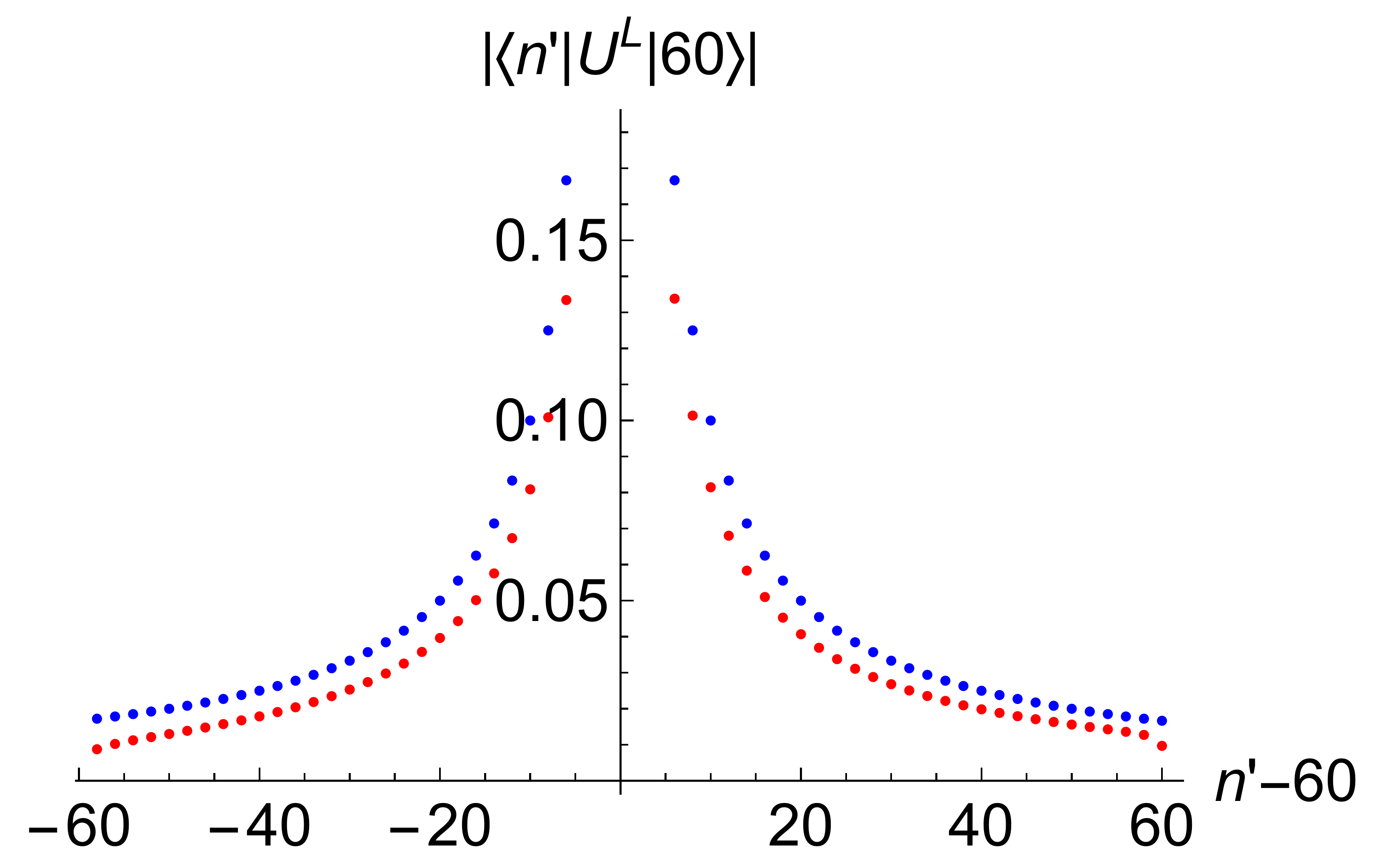}
\end{center}
\caption{A more complete presentation of the off-diagonal matrix elements of the perturbation (again represented, in the case of a parametric excitation, by the potential energy).
This figure is supplementary to the one presented in Fig.~\ref{f:matrix_elements}(a). Here, the full information about the sign the quantum number differences $n'-n$ is given.  Only even values of $n'-n$ are represented, since the odd values vanish due to parity. The red dots represent numerical results. The $\ln(n)$-spectrum potential used had $N_{\text{b}}=120$
bound states. We also present the theoretical prediction \eqref{off-diagonal_semiclassics_no_fit} that contains no fit parameters (blue dots).   
}
\label{f:off-diagonal_FULL_INFO}
\end{figure}
%

\subsection{Diagonal matrix elements of the potential energy}
For completeness, we present the semiclassical approximation for the diagonal matrix elements of the potential energy. Again, within the accuracy of the semiclassical approximation, we can replace the $\ln(n)$-spectrum potential Eq.~\eqref{log_n} with its 
classical counterpart Eq.~\eqref{log_n_potential_classical}.

Generally, diagonal matrix elements of observables are approximated using classical temporal averages. In our case however, a virial theorem that has been previosly applied to a logarithmic potential in Ref.~\cite{mack2010_032119} allows us to avoid the temporal average integrals. Using the result obtained in \cite{mack2010_032119}, we get: 
\begin{align}
\langle n | U^{\text{L}} | n \rangle \approx U_{0} \ln(e^{-\frac{1}{2}} n)
\label{diagonal_semiclassics}
\,\,.
\end{align}
Fig.~\ref{f:matrix_elements}(b) shows an excellent agreement between our prediction and \emph{ab initio} results for a $\ln(n)$-spectrum potential with $N_{\text{b}} = 120$ bound states. 

\section{The exponential lattice \label{s:exponential_lattice}}
The goal of this appendix is to develop a theory for a quantum infinite one-dimensional lattice whose hopping coefficients decay exponentially in space. This model is relevant to the problem of a parametric excitation of the $\ln(n)$-spectrum potential. In particular, using the exponential lattice model, we were able to identify the source of a localization in the space of the unperturbed eigenstates and devise the ways to break such a localization.

Consider the following one-body Hamiltonian:
\begin{align}
\hat{H}_{\text{EL}} =  - J_{0}\sum_{m=-\infty}^{+\infty} e^{-\gamma m}\,\left( |m+1\rangle\langle m| + |m\rangle\langle m+1|   \right)
\,\,,
\label{H_EL}
\end{align}
where $m$ are the indices of the lattice sites, and  $\gamma$ is a positive constant. Below, we will discuss the solutions to the time-independent Schr\"{o}dinger equation, 
\begin{align}
\begin{split}
&
-J_{0} e^{-\gamma m} \left(e^{\gamma} \psi_{m-1} +  \psi_{m+1}\right) = E \psi_{m}
\\
&
m = 0,\,\pm 1,\,\pm 2,\,\ldots
\end{split}
\,\,,
\label{Schrodinger_EL}
\end{align}
where $\psi_{m}$ is the coordinate representation of an eigenstate 
\begin{align*}
|\psi\rangle = \sum_{m=-\infty}^{+\infty} \psi_{m} \, |m\rangle
\end{align*}
with an energy $E$
\begin{align*}
\hat{H}_{\text{EL}} |\psi\rangle = E |\psi\rangle
\,\,.
\end{align*}

For reasons that will become clear below, we have to start our discussion not from the eigenenergies but from the  
$m\to -\infty$ behavior of the eigenstates. We will proceed to the analysis of the spectrum, followed by by the $m\to +\infty$ asymptotic behavior. 

\subsection{Eigenfunctions in the ``dark state'' region, $m\to -\infty$ \label{ss:dark_states}}
Consider an eigenstate $|\psi\rangle$ with an energy $E$. At  
\[m \ll \frac{1}{\gamma} \ln\left(\frac{J_{0}}{|E|}\right) \,\,, \]
where $|E| \ll J_{m}$, the coefficients in front of the wavefunction components $\psi_{m\pm1}$ 
in the Schr\"{o}dinger equation Eq.~\eqref{Schrodinger_EL}
exponentially explode, reducing the problem to 
finding the kernel, i.e. the ``dark states'' (DS) \cite{harris1993_552,arimondo1996_257}, of the Hamiltonian in Eq.~\eqref{H_EL}, i.e., where 
\begin{align}
\hat{H}_{\text{EL}} |\psi_{\text{DS}}\rangle = 0
\,\,,
\label{EL_kernel}
\end{align}
or
\begin{align}
\begin{split}
&
e^{\gamma} (\psi_{\text{DS}})_{m-1} +  (\psi_{\text{DS}})_{m+1} = 0,
\\
&
m = m_{\text{DS}},\,m_{\text{DS}}-1,\,m_{\text{DS}}-2,\,\ldots
\end{split}
\,\,
\label{Schrodinger_EL_DS}
\end{align}
Here and below, $m_{\text{DS}}$ is the right boundary of the zone of validity of the ``dark state'' approximation in Eq.~\eqref{Schrodinger_EL_DS}, and $J_{m} \equiv J_{0} e^{-\gamma m}\,\,$.

The wave function $(\psi_{\text{DS}})_{m}$ can be immediately found to be  
\begin{align}
\begin{split}
&
\psi_{m} \stackrel{m\to -\infty}{\approx} (\psi_{\text{DS}})_{m}
\\
&
(\psi_{\text{DS}})_{m} = 
\left\{
\begin{array}{ccc}
(-1)^{\frac{m_{\text{DS}}-m}{2}} e^{-\gamma \frac{m_{\text{DS}}-m}{2}} \psi_{m_{\text{DS}}} 
&\text{for}&  m-m_{\text{DS}} = \text{even}
\\
0
&\text{for}&  m-m_{\text{DS}} = \text{odd} 
\end{array}
\right.
\\
&
m = m_{\text{DS}},\,m_{\text{DS}}-1,\,m_{\text{DS}}-2,\,\ldots,
\end{split}
\label{psi_EL_DS}
\end{align}
where $\psi_{m_{\text{DS}}}$ is the value of the {\it exact} eigenfunction at $m_{\text{DS}}$ and is presumed to be known. Expectedly, in terms of lattices with constant hopping coefficients, this state can be interpreted a superposition of the two zero energy plane waves with momenta $\pm \pi/2$. 

Note the following however. As $m$ tends to $-\infty$, the wavefunction in Eq.~\eqref{psi_EL_DS} decays as $e^{\gamma \frac{m}{2}}$,  while the matrix elements of the Hamiltonian in Eq.~\eqref{H_EL} explode more quickly, i.e.\ as $e^{-\gamma m}$. The only reason why the energy of the eigenstate does not diverge is the destructive interference between the two terms in Eq.~\eqref{Schrodinger_EL_DS}. 

Imagine now that one decides to truncate the lattice at some negative position $m_{\text{left}}$, in such a way that the allowed values of the coordinate $m$ now span the range 
$m = m_{\text{DS}},\,m_{\text{DS}}-1,\,m_{\text{DS}}-2,\,\ldots,\, m_{\text{l}}+1,\, m_{\text{left}}$. The norm of the wavefunction 
in Eq.~\eqref{psi_EL_DS} converges to a finite value for $m_{\text{left}} \to -\infty$. Hence in this limit, the wavefunction 
remains finite. However, in the same limit, the left hand side of Eq.~\eqref{EL_kernel} will exponentially explode,
\begin{widetext}
\begin{align*}
\hat{H}_{\text{EL}} |\psi_{\text{DS}}\rangle = 
-J_{0} e^{-\gamma m_{\text{left}}} 
(-1)^{\frac{m_{\text{DS}}-(m_{\text{left}}+1)}{2}} e^{-\gamma \frac{m_{\text{DS}}-(m_{\text{left}}+1)}{2}} \psi_{m_{\text{DS}}} 
\propto e^{+\gamma \frac{|m_{\text{left}}|}{2}} 
\,\,,
\end{align*}
\end{widetext}
unless $m_{\text{left}}$ and $m_{\text{DS}}$ have the same parity, implying that $m_{\text{DS}}-m_{\text{left}}$ is even.

We arrive at the following conclusion: in the problem at hand, extending our lattice to $m=-\infty$ may produce unphysical eigenstates. Hence, we suggest the following amendment to our model: we will continue being interested in the eigenstates that are localized far away from the left boundary $m_{\text{left}}$, rendering the actual position of this boundary irrelevant. However, the {\it parity} of $m_{\text{left}}$ will remain important. Without loss of generality, we may assume that $m_{\text{left}}$ is even. 
Finally, using the fact that $m_{\text{DS}} - m_{\text{left}}$ must be even, we establish the following rule: the infinite lattice solutions of the Schr\"{o}dinger 
equation in Eq.~\eqref{Schrodinger_EL} must be {\it post-selected} in such a way that only the eigenstates with the ``dark state'' asymptotic behavior Eq.~\eqref{psi_EL_DS} where $m_{\text{DS}}$ is even are kept. This rule plays an important role in the next subsection. 

\subsection{The $+E \leftrightarrow -E$ symmetry, translational invariance, and the energy spectrum}
Consider the Schr\"{o}dinger equation in Eq.~\eqref{Schrodinger_EL}. Two properties of the spectrum can be proven.

\begin{property}[Positive-negative energy parity]
\label{prop:sign_change}
Let $\psi^{(+)}_{m}$ be a solution of the eigenvalue problem Eq.~\eqref{Schrodinger_EL} corresponding to an eigenenergy
$E=E^{(+)}$. Then the eigenstate-eigenenergy pair 
\begin{align}
\begin{split}
&
\psi^{(-)}_{m} = (-1)^{m} \psi^{(+)}_{m}
\\
&
E=E^{(-)}=-E^{(+)} 
\end{split}
\,\,.
\label{EL_energy_sign_change}
\end{align}
is also a solution of Eq.~\eqref{Schrodinger_EL}. The statement will remain valid if the infinite lattice is reduced to a ray or a line segment bounded by a Dirichlet boundary condition. 
\end{property}
\begin{proof}
Indeed, substituting $\psi^{(-)}_{m}$ to Eq.~\eqref{Schrodinger_EL}, we get 
\begin{align*}
&
-J_{0} e^{-\gamma m} \left(e^{\gamma} \psi^{(-)}_{m-1} +  \psi^{(-)}_{m+1}\right) 
\\
&
\stackrel{\eqref{EL_energy_sign_change}}{=} -J_{0} e^{-\gamma m} \left(e^{\gamma} (-1)^{m-1} \psi^{(+)}_{m-1} +  (-1)^{m+1}\psi^{(+)}_{m+1}\right) 
\\
&
=-J_{0} e^{-\gamma m} (-1) (-1)^{m} \left(e^{\gamma} \psi^{(+)}_{m-1} + \psi^{(+)}_{m+1}\right) 
\\
&
\stackrel{\eqref{Schrodinger_EL}}{=} (-1) (-1)^{m} E^{(+)} \psi^{(+)}_{m} 
\\
&
\stackrel{\eqref{EL_energy_sign_change}}{=} E^{(-)} \psi^{(-)}_{m}
\,\,.
\end{align*}
\end{proof}
\begin{property}[Translational invariance]
\label{prop:translational}
Let $\psi^{(\Delta m=0)}_{m}$ be a solution of the eigenvalue problem Eq.~\eqref{Schrodinger_EL} corresponding to an eigenenergy
$E=E^{(\Delta m=0)}$. Then the eigenstate-eigenenergy pair 
\begin{align}
\begin{split}
&
\psi^{(\Delta m)}_{m} = \psi^{(\Delta m=0)}_{m-\Delta m}
\\
&
E=E^{(\Delta m)}= e^{-\gamma \Delta m} E^{(\Delta m=0)} 
\end{split}
\,\,.
\label{EL_translational_invariance}
\end{align}
is also a solution of Eq.~\eqref{Schrodinger_EL}. 
\end{property}
\begin{proof}
Again, substituting $\psi^{(\Delta m)}_{m}$ to Eq.~\eqref{Schrodinger_EL}, we get 
\begin{align*}
&
-J_{0} e^{-\gamma m} \left(e^{\gamma} \psi^{(\Delta m)}_{m-1} +  \psi^{(\Delta m)}_{m+1}\right) 
\\
&
\stackrel{\eqref{EL_translational_invariance}}{=} -J_{0} e^{-\gamma m} \left(e^{\gamma} \psi^{(\Delta m=0)}_{m-\Delta m-1} +   
\psi^{(\Delta m=0)}_{m-\Delta m+1}\right) 
\\
&
= -J_{0} e^{-\gamma\, \Delta m}  e^{-\gamma (m-\Delta m)} \left(e^{\gamma} \psi^{(\Delta m=0)}_{m-\Delta m-1} +   
\psi^{(\Delta m=0)}_{m-\Delta m+1}\right) 
\\
&
\stackrel{\eqref{Schrodinger_EL}}{=} e^{-\gamma\, \Delta m}  E^{(\Delta m=0)} \psi^{(\Delta m=0)}_{m-\Delta m} 
\\
&
\stackrel{\eqref{EL_translational_invariance}}{=} E^{\Delta m} \psi^{(\Delta m)}_{m}
\,\,.
\end{align*}
\end{proof}

The Property\! \ref{prop:translational} may seem paradoxical at first, as at may seem to imply, in combination with the Property\!
\ref{prop:sign_change} that there are twice as many eigenstates of the Hamiltonian as there are lattice sites. This apparent paradox is resolved in Subsection \ref{ss:dark_states}. As we show there, a model given by Eq.~\eqref{H_EL} on an infinite lattice is not physical, as it is not a limit of a problem on a ray of sites to the right of a Dirichlet boundary. 
However an infinite lattice model where the eigenstates are post-selected in such a way that only the states 
that in the limit $m\to -\infty$ develop nodes at the {\it odd} sites are kept represents a faithful 
\begin{align*}
&
m_{\text{left}} \to -\infty
\\
&
m_{\text{left}} = \text{even}
\end{align*}
limit of a ray with a Dirichlet boundary at $m_{\text{left}}$. More precisely, 
we assume that $\psi_{m_{\text{left}}-1} = 0$ at the $(m_{\text{left}}-1)^{\text{th}}$ site. This consideration leads to the following amendment to the 
Property\! \ref{prop:translational}:
\begin{align}
\Delta m = \text{even}
\,\,.
\label{EL_translational_invariance_supplement}
\end{align}
%

\subsection{Eigenfunctions in the ``classically forbidden'' region, $m\to +\infty$}
In this subsection, we will be interested in the region of space where 
\begin{equation}
|E| \gg J_{m},
\label{classically_forbidden_condition} 
\end{equation}
or
\begin{equation}
m \gg \frac{1}{\gamma} \ln\left(\frac{J_{0}}{|E|}\right)
\nonumber.
\end{equation}

If one forgets about the spatial dependence of $J_{m}$ for a moment, this area of space corresponds to a ``classically forbidden'' (CF) region where no waves can propagate. One would expect the wavefunction to develop an exponentially-decaying tail instead.   Furthermore, we will assume that the spatial variation of the hopping coefficient is slow in comparison with the spatial variation of the eigenstates. This is a domain of the parameter space where a semiclassical approximation can be used~\cite{landau_quantum___quasiclassics}. We will verify the validity of this assumption \emph{a posteriori}. 

Consider first the case where $E>0$. If the hopping coefficient were a constant, $\bar{J}$, then the exponentially decaying solution of the Schr\"{o}dinger equation Eq.~\eqref{Schrodinger_EL} would have the form
\begin{align}
\begin{split}
&
(\psi_{\text{$J_{m} = \bar{J}$}})_{m} = 
(-1)^{m-m_{\text{CF}}}
e^{-\bar{\kappa}m}
\psi_{m_{\text{CF}}} 
\end{split}
\,\,,
\label{psi_EL_CF_J_constant}
\end{align}
where $m_{\text{CF}}$ is the (left) boundary of the area of validity of the CF region approximation, 
$\psi_{m_{\text{CF}}}$ is the value of the {\it exact} eigenfunction at $m_{\text{CF}}$, and
\[\bar{\kappa} = \text{arccosh}\!\left(\frac{E}{2 \bar{J}}\right) \stackrel{|E| \gg \bar{J}}{\approx} \ln\!\left(\frac{E}{\bar{J}}\right) \,\,.  \]

To first order of the semiclassical approximation \cite{landau_quantum___quasiclassics}, the $\bar{\kappa} m$ under the exponent in \eqref{psi_EL_CF_J_constant} is modified as 
\[ \bar{\kappa} m \to \sum_{m'=m_{\text{CF}}}^{m-1} \kappa_{m'} \,\,,\]
with $\kappa_{m} \equiv \ln\!\left(\frac{E}{J_{m}}\right)$. After straightforward manipulation (see App.~\ref{ss:WKB} for details), we get the following 
$m\to +\infty$ approximation, still restricted to positive energies, $E>0$:  
\begin{align*}
\begin{split}
&
\psi_{m} \stackrel{m\to +\infty}{\approx} (\psi_{\text{CF}})_{m}
\\
&
(\psi_{\text{CF}})_{m} = 
(-1)^{m-m_{\text{CF}}}
\left(\frac{\sqrt{J_{m_{\text{CF}}} J_{m-1}}}{E}\right)^{m-m_{\text{CF}}}
\psi_{m_{\text{CF}}} 
\\
&
m = m_{\text{CF}},\,m_{\text{CF}}+1,\,m_{\text{CF}}+2,\,\ldots
\end{split}
\,\,.
\end{align*}

The case of negative energies, $E<0$, can be evaluated in a similar fashion. The only difference between the 
two signs of energies is that at the negative energies, there is no sign-alternating factor in the expression analogous to 
Eq.~\eqref{psi_EL_CF_J_constant}. 

Finally, an expression that covers both signs of energy reads 
\begin{widetext}
\begin{align}
\begin{split}
&
\psi_{m} \stackrel{m\to +\infty}{\approx} (\psi_{\text{CF}})_{m}
\\
&
(\psi_{\text{CF}})_{m} = 
\left\{
\begin{array}{ccc}
(-1)^{m-m_{\text{CF}}} & \text{for} & E>0
\\
1 & \text{for} & E<0
\end{array}
\right\}
\left(\frac{\sqrt{J_{m_{\text{CF}}} J_{m-1}}}{|E|}\right)^{m-m_{\text{CF}}}
\psi_{m_{\text{CF}}} 
\\
&
m = m_{\text{CF}},\,m_{\text{CF}}+1,\,m_{\text{CF}}+2,\,\ldots
\end{split}
\,\,.
\label{psi_EL_CF}
\end{align}
\end{widetext}

Let us finally assess the validity of the semiclassical approximation. According to Ref.~\cite{landau_quantum___quasiclassics}, the 
semiclassical approximation is valid when the de-Broglie wavelength, $\lambda_{\text{d-B}}$, does not change appreciably over a length comparable to itself: 
\[ \frac{d \lambda_{\text{d-B}}}{dm} \ll 1 \,\,.\]
In our case, 
\[\lambda_{\text{d-B}} \sim \frac{1}{\kappa_{m}} \sim  \frac{1}{\gamma (m-\tilde{m})} \,\,, \]
with $\tilde{m} = -\frac{1}{\gamma} \ln\!\left(\frac{|E|}{J_{0}}\right)$ is the location where the hopping coefficient is comparable to the magnitude of energy. The condition of validity of the approximation becomes
\begin{align}
|E| \gg e^{\sqrt{\gamma}} J_{m}
\label{semiclassics_validity_condition} 
\,\,.
\end{align}
Notice that this condition is more restrictive than the condition in Eq.~\eqref{classically_forbidden_condition} for the CF region. 

\subsection{Derivation of Eq.~\eqref{psi_EL_CF} for the eigenstate wavefunctions in the CF region \label{ss:WKB}}

In this subsection we will derive the WKB approximation for the exponential lattice which leads to Eq.~\eqref{psi_EL_CF}. We will follow Bremmer's method as described in Ref.~\cite{Berry_1972}.
Since in the CF region, the spatial dependence of the hopping coefficient is slow, we divide the domain of interest into regions, not necessarily of the same size, with an approximately constant $J_n$ in each region. The right boundary of each region is labeled by $n_j$, and the difference in magnitudes between the $J_j$ in neighboring regions is small.  

We consider the scattering problem with a step at $n_j$. The hopping coefficient to the left is $J_L$ and to the right is $J_R$. To the right of the step the wavefunction is a decaying exponential with wavenumber $\kappa_R$, and to the left $\kappa_L$. The transmission coefficient is found to be
\begin{equation*}
t_{n_j} = \frac{J_L(e^{\kappa_L} - e^{-\kappa_L})}{J_L e^{\kappa_L}-J_R e^{-\kappa_R}}e^{(\kappa_R - \kappa_L)n_j}.
\end{equation*}

To evaluate the wavefunction at some later point, we neglect reflection coefficients, which are very small in magnitude. We are free to choose the size of the regions in our domain, and more importantly their endpoints, so we may choose the endpoint of the final region, $n_N$, to be at the lattice point where we wish to evaluate the wavefunction. There, the wavefunction is approximated by
\begin{align*}
&
\psi_{n_N} = e^{-\kappa_{n_{N+1}}n_N} \times
\\
&
\qquad
\prod_{j=0}^{N-1}
\frac{J_{n_j}(e^{\kappa_{n_j}} - e^{-\kappa_{n_j}})}{J_{n_j} 
e^{\kappa_{n_j}}-J_{n_{j+1}}
e^{-\kappa_{n_{j+1}}}}e^{(\kappa_{n_{j+1}} - \kappa_{n_j})n_j}.
\end{align*}

Since the difference between neighboring $J_n$ is small,
\begin{equation*}
\frac{J_{n_j}}{J_{n_{j+1}}} = 1+\epsilon_j, \qquad \abs{\epsilon_j} \ll 1.
\end{equation*}
Then, to first order in $\epsilon_j$,
\begin{equation*}
\frac{J_{n_j}(e^{\kappa_{n_j}} - e^{-\kappa_{n_j}})}{J_{n_j} e^{\kappa_{n_j}}-J_{n_{j+1}}e^{-\kappa_{n_{j+1}}}} =  \sqrt{\frac{J_{n_j}\sinh{\kappa_{n_j}}}{J_{n_{j+1}}\sinh{\kappa_{n_{j+1}}}}},
\end{equation*}
which becomes a telescoping series when the product is taken, leaving

\begin{equation*}
\psi_{n_N} = \sqrt{\frac{J_{n_0}\sinh{\kappa_{n_0}}}{J_{n_{N-1}}\sinh{\kappa_{n_{N-1}}}}}  e^{-\kappa_{n_{N+1}}n_N}\prod_{j=0}^{N-1} e^{(\kappa_{n_{j+1}} - \kappa_{n_j})n_j}.
\end{equation*}
The remaining exponential terms can be rewritten in terms of a sum, so that
\begin{equation*}
\psi_{n_N} = \sqrt{\frac{J_{n_0}\sinh{\kappa_{n_0}}}{J_{n_{N-1}}\sinh{\kappa_{n_{N-1}}}}} e^{-n_0\kappa_0+\sum_{j=1}^{N-1} (n_{j-1}-n_j)\kappa_j}.
\end{equation*}
We can now restore the $\kappa_j$ associated with particular lattice points, so the WKB approximation for the wavefunction at a particular lattice point $m$ becomes
\begin{equation*}
\psi_m = \sqrt{\frac{J_{m_0}\sinh{\kappa_{m_0}}}{J_{m-1}\sinh{\kappa_{m-1}}}} e^{-m_0\kappa_0-\sum_{j=1}^{m-1}\kappa_j}.
\end{equation*}

Next we invoke Eq.~\eqref{classically_forbidden_condition}, the condition that $\abs{E}\gg J_m$. We consider propagation only within the classically forbidden region, so $m_0=m_{cf}$, and the summation in the exponent is now
\begin{equation*}
\sum_{j=m_{CF}}^{m-1} \kappa_j.
\end{equation*}
To first order in $J_m/E$, this reduces the expression for the wavefunction to
\begin{equation*}
\qty(\psi_{CF})_m = e^{-m_{CF}\kappa_{CF}}\qty(\frac{J_0}{\abs{E}}e^{-\gamma (m_{CF}+m-1)/2})^{m-m_{CF}},
\end{equation*}
which can be rewritten using the definition of $J_m$ to read
\begin{equation}
\qty(\psi_{CF})_m  = \qty(\frac{\sqrt{J_{m_{CF}}J_{m-1}}}{\abs{E}})^{m-m_{CF}}\psi_{m_{CF}}.
\end{equation}
 
\subsection{Numerically exact eigenstates and discussion}
Figures \ref{f:exponential_lattice__spectrum} and \ref{f:exponential_lattice__eigenstates} show results numerical results of a 
numerical diagonalization of the Hamiltonian in Eq.~\eqref{H_EL} on a lattice with $201$ sites $m=0,\,\pm1,\,\pm 2,\,\ldots,\,\pm 100$ with Dirichlet boundary conditions. The parity of the coordinate of the leftmost site (i.e., the fact that $-100$ is an even number) effectuates the post-selection condition that $m_{\text{DS}} - m_{\text{left}}$ must be even. This can be reformulated as follows: our numerical procedure preserves only half of the eigenstates of the infinite lattice, namely those whose nodes in the $m\to -\infty$ region are located on the {\it odd} lattice sites.

We considered two values of the coupling decay rate $\gamma$: $\gamma = 0.3$ and $\gamma = \ln(3) \approx 1.10$. In the first case, coupling decay is slow enough to ensure the validity of the approximation in Eq.~\eqref{psi_EL_CF}. The second case corresponds to a lattice that is relevant to the problem in the main text. Fig \ref{f:exponential_lattice__spectrum} shows energy spectra in both the $\gamma = 0.3$ and $\gamma = \ln(3)$ cases. We noticed that in both cases, values $E= +J_{0}$ and $E= -J_{0}$ belong to the spectrum. We were not able to support this observation with an analytic result.

Next, we used the states $E= +J_{0}$ and $E= -J_{0}$ as a reference to verify the predictions in Eqs.~\eqref{EL_translational_invariance}-\eqref{EL_translational_invariance_supplement}. In both cases, the agreement between the numerical results and Eqs.~\eqref{EL_translational_invariance}-\eqref{EL_translational_invariance_supplement} is remarkable. We also confirmed the prediction that the spectrum is symmetric with respect to the $+E \leftrightarrow -E$ transformation in Eq.~\eqref{EL_energy_sign_change}. Figure~\ref{f:exponential_lattice__eigenstates} shows the  $E= +J_{0}$ and $E= -J_{0}$ eigenstates for both values of $\gamma$ considered here. Both the DS (Eq.~\eqref{psi_EL_DS}) and the CF region (Eq.~\eqref{psi_EL_CF}) approximations work remarkably well. The validity of the approximation in Eq.~\eqref{psi_EL_CF} for the 
$\gamma = \ln(3)$ case was unexpected, given the validity condition in Eq.~\eqref{semiclassics_validity_condition}.  
\begin{figure}[!h]
\includegraphics[width=.4\textwidth]{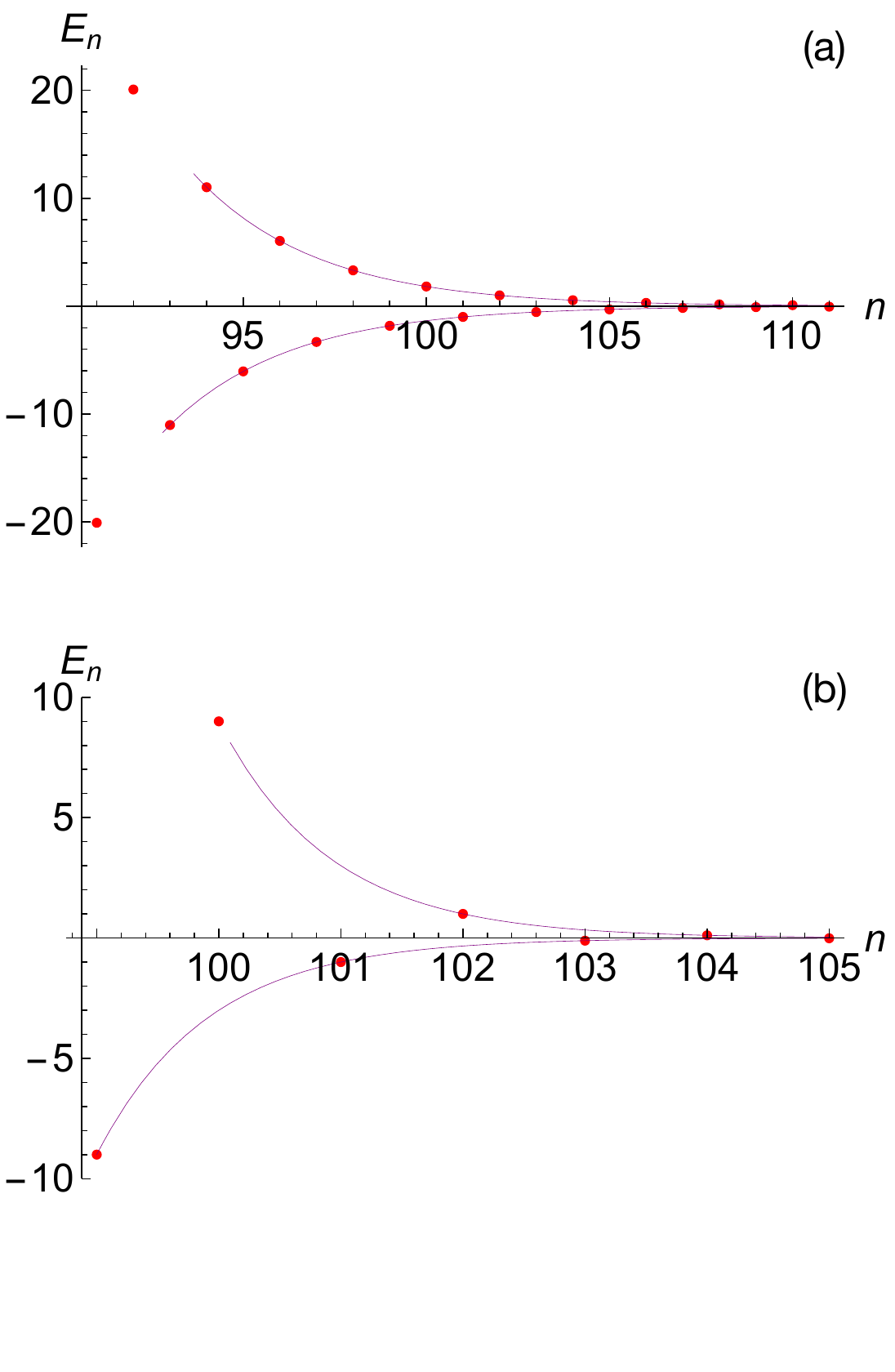}
\caption{Exponential lattice: energy spectra. The coupling constant decay rate is given by (a) $\gamma = 0.3$ and (b) $\gamma = \ln(3)$ 
respectively. Spectra are ordered in the descending order of the magnitude of energy, and energy is measured in the units of $J_{0}$. The red dots show the numerical prediction, and the solid thin lines represent the predictions in Eqs.~\eqref{EL_translational_invariance}-\eqref{EL_translational_invariance_supplement}. Notice also that for any magnitude of eigenenergy, both signs of the eigenenergy are present in the spectrum, in agreement with Eq.~\eqref{EL_energy_sign_change}. Finally, the presence of the $E=\pm J_{0}$ levels for both values of $\gamma$ remains unexplained.}
\label{f:exponential_lattice__spectrum}

\end{figure}
\begin{figure*}

\includegraphics[width=.9\textwidth]{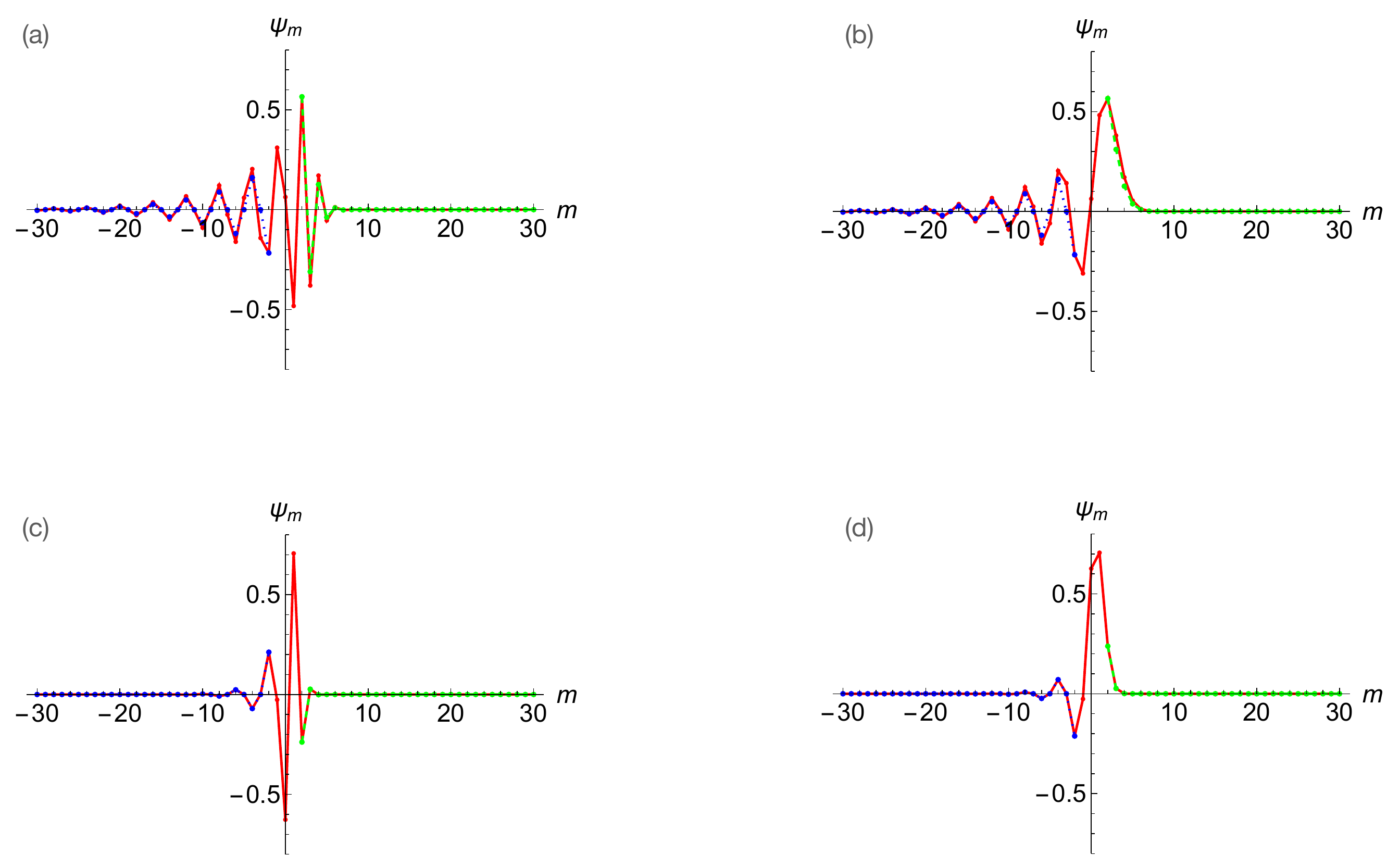}

\caption{Exponential lattice eigenstates. 
The $E=+J_{0}$ ((a) and (c)) and $E=-J_{0}$ ((b) and (d)) eigenstates of the $\gamma = 0.3$  ((a) and (b)) and 
$\gamma = \ln(3)$  ((c) and (d)) lattices. The rest of the eigenstates are translations of the ones presented by an even number of sites. Red circles and red solid lines show numerical results. Blue circles and blue dotted lines correspond to the DS approximation in Eq.~\eqref{psi_EL_DS}. Green circles and green dashed lines correspond to the CF region approximation in Eq.~\eqref{psi_EL_CF}.   
}
\label{f:exponential_lattice__eigenstates}
\end{figure*}
%


\section{A conjecture that there exists an infinite number of even numbers whose Goldbach decomposition does not involve lower twin primes \label{s:twin_primes}}
In this section, we conjecture that there is an infinite number of even numbers $w$ that feature the No Lower Twin (NLT) property, i.e., that for any of the Goldbach decompositions of $w$, $w=p_{1}+p_{2}$, both $p_{1}+2$ and $p_{2}+2$ are composite numbers. 
Note that unlike in the main text, here we do \emph{not} assume that $p_{1}$ and $p_{2}$ are in any particular order with respect to each other. 

In preparation, let us reiterate a well-known result that there do not exist triplets of consecutive (separated by $2$) primes unless the lower member of the triplet is $3$. We will need, however, a stronger version of this statement expressed in the following Lemma.
\begin{lemma}
Any lower twin prime $p$ greater than $3$ the following is true:
\[p \equiv 2 \,\,\, (\text{\rm mod}\,\, 3)\]
\label{p_mod_2}
\end{lemma}
\begin{proof}
First of all, since $p$ is a prime different from $3$, 
\[p \equiv 1 \text{\rm \,\,or } 2 \,\,\, (\text{\rm mod}\,\, 3)\,\,.\]
Secondly, since $p+2$ is also a prime different from $3$,
\[p \equiv 2 \text{\rm \,\,or } 0 \,\,\, (\text{\rm mod}\,\, 3)\,\,.\]
Combining the two we get, 
\[p \equiv 2 \,\,\, (\text{\rm mod}\,\, 3)\,\,.\]
\end{proof}

We will also need to assert that
\begin{lemma}
All the even values of $w$ such that 
\begin{align*}
&
\text{\rm (a)  } w-3 \neq \text{\rm prime}
\\
&
\text{\rm (b)  } w\equiv 2\,\,\, (\text{\rm mod}\,\, 6)
\end{align*}
possess the NLT property. 
\label{benz_s_lemma}
\end{lemma}
For example:
\begin{align*}
&
38 = 6\times 6 + 2 = 2 \mod 6
\\
&
38-3=35 = 5\times 7 \neq \text{prime}
\\
&
38 = 7 + 31 = 19 + 19
\\
&
7+2=9=3\times 3 \neq \text{prime}
\\
&
31+2 = 33 = 3\times 11 \neq \text{prime}
\\
&
19+2 = 21 = 3\times 7 \neq \text{prime}
\end{align*}
\begin{proof}
Consider a Goldbach partition of $w$, 
\[w=p_{1}+p_{2}\,\,.\]

From the premise (a) it follows that 
\begin{align}
\begin{split}
&
p_{1} \neq 3
\\
&
p_{2} \neq 3
\end{split}
\label{neq_3}
\,\,.
\end{align}
This can be proven using  {\it reductio ad absurdum}. Assume that $p_{2}=3$. Then $p_{1}=w-p_{2} = (w-3)+(3-p_{2})$ is a composite number, in contradiction to the premise that $p_{1}$ was a prime number. The assertion that $p_{1} \neq 3$ can be proven analogously.

Next, it will follow from the premise (b) that  
\begin{align}
w\equiv 2\,\,\, (\text{\rm mod}\,\, 3)
\label{w_mod_3}
\end{align}

We will utilize the statement in Eq.~\eqref{w_mod_3}, Lemma~\ref{p_mod_2} to build a proof {\it ad absurdum}. Assume that $p_{1}$ is a lower twin different from $3$. Then
\begin{align*}
p_{2} = w - p_{1} 
\equiv 2-2  \,\,\, (\text{\rm mod}\,\, 3) 
\equiv 0  \,\,\, (\text{\rm mod}\,\, 3)
\end{align*}
However, this statement contradicts the premise that $p_{2}$ is prime different from $3$. 
\end{proof}

We are finally ready to the principal theorem of this section:
\begin{theorem}
\label{th:non-twin}
The number of even numbers possessing the NLT property is infinite.
\end{theorem}
\begin{proof}
Consider a sequence 
\begin{align}
\begin{split}
&
w_{k} \equiv 2(15k+4)
\\
&
k=1,\,2,\,3,\,4,\,\ldots
\end{split}
\,\,.
\label{provable_sequence}
\end{align}
Let us first show that every member of this sequence obeys the premise (a) of the Lemma~\ref{benz_s_lemma}. Indeed, \[2(15k+4)-3 = 5(6k+1)\,\,;\]
the right hand side is a manifestly composite number unless $k=0$
For the premise (b), we get \[2(15k+4) = 6k' + 2\,\,,\]
with $k'=5k+1$.
\end{proof}

A comment is in order. The members of the infinite sequence \eqref{provable_sequence} used to prove the Theorem~\ref{th:non-twin}, may turn out to be a small subset of the full set of evens possessing the NLT property. Let us attempt to make an estimate of the density of evens in the former set. According to the Prime Number Theorem, the number of primes below a number $N$ is, approximately, $\pi(N)\approx N/\ln(N)$~\cite{dudley_Number_Theory1970}.  It will immediately follow that for a given prime $p_{j}$, the next prime $p_{j+1}$ will be separated from $p_{j}$ by a gap---filled by the composite numbers---whose size is approximately $\ln(N)$. (Here and below, $p_{1},\,p_{2},\,p_{3},\,\ldots = 2,\,3,\,5$ is the contiguous sequence of prime numbers in ascending order.) It will also follow that $p_{j} \approx j \ln(j)$.

The probability that a given number $N$ obeys the premise (b) of Lemma~\ref{benz_s_lemma} is $1/6$. The probability that a number of this type obeys the premise (a) 
is $1-2/\ln(N)$. Here $1/\ln(N)$ is the inverse of the gap between consecutive primes; the factor of $2$ accounts for the fact that premise (b) implies that $N$ is even and hence that $N-3$ is odd. Now, according to Lemma~\ref{benz_s_lemma}, the probability that a given number $N$ is an even number that possesses the NLT property is $(1/6)(1-2/\ln(N)) \stackrel{N\gg 1}{\approx} 1/6$. All in all, we estimate that number of the even numbers obeying the NLT property that are less than $N$ is 
\begin{align}
\pi_{\text{NLT, estimated}}(N) \stackrel{N\gg 1}{\approx} \frac{N}{6} \,\,. 
\label{NLT_estimate}
\end{align}

Interestingly, the rigorously justified subsequence \eqref{provable_sequence} properly captures the uniformity of the distribution estimated by Eq.~\eqref{NLT_estimate} but underestimates the density of the NLT even numbers by a factor of $5$:
\begin{align}
\pi_{\text{NLT, proven}}(N) \geq  \frac{N}{30} - \frac{19}{15} \,\,. 
\label{NLT_lower_bound}
\end{align}
%

\bibliography{Bethe_ansatz_v048,Nonlinear_PDEs_and_SUSY_v050,UoB_bib,WKBRef,BensRefs}
\end{document}